\documentclass[letterpaper,11pt]{article}

\usepackage[left=1in, right=1in, top=1in, bottom=1in]{geometry}

\usepackage{amsfonts,amsmath,amssymb,amsthm,nicefrac}
\usepackage{algorithmic}
\usepackage{authblk}
\usepackage{mathtools,thmtools,nicefrac}
\usepackage{newpxtext}
\usepackage{newpxmath}
\usepackage[utf8]{inputenc}
\usepackage[T1]{fontenc}
\usepackage{microtype}
\usepackage{tikz}
\usepackage{xcolor}
\usepackage{lipsum}
\usepackage{xspace}
\usepackage{csquotes}
\usepackage{listings}
\usepackage{float}
\usepackage{setspace}
\usepackage{tabularx}
\usepackage{booktabs} 
\usepackage{tabulary}
\usepackage{threeparttable}
\usepackage{etoolbox}
\usepackage{multirow}
\usepackage[style=trad-abbrv,
citestyle=numeric-comp,
sorting=nyt,
sortcites,
backend=biber,
maxbibnames=99,
maxcitenames=7,
mincitenames=1,
maxalphanames=4,
minalphanames=3,
doi=false
]{biblatex}
\DefineBibliographyStrings{english}{%
    andothers = {et\addabbrvspace al\adddot}
}

\def\nobreakbefore{\relax
  \ifvmode\else
    \ifhmode
      \ifdim\lastskip > 0pt\relax
        \unskip\nobreakspace
      \fi
    \fi
  \fi
}
\let\oldcite\cite
\renewcommand\cite{\nobreakbefore\oldcite}

\defcounter{abbrvpenalty}{9000}

\DeclareSourcemap{
  \maps[datatype=bibtex]{
    \map{
      \step[fieldsource=url,
            match=doi,
            final]
      \step[fieldset=url, null]
    }
  }
}

\AtEveryBibitem{
 \ifentrytype{inproceedings}{%
  \clearfield{series}
  \clearfield{number}
  \clearfield{volume}
 }{}%
   \clearfield{url}%
 \clearlist{address}%
 \clearfield{issn}%
 \clearlist{location}%
 \clearfield{month}%
 \ifentrytype{book}{}{
  \clearlist{publisher}%
  \clearname{editor}%
 }%
}

\usepackage{tcolorbox}
\newcommand*{\graybox}[2]{\raisebox{0.5em}{\colorbox{gray}{\parbox{#1}{\centering\color{white}#2}}}}

\usepackage[colorlinks=true,citecolor=purple,linkcolor=blue,urlcolor=gray]{hyperref}
\usepackage[capitalize]{cleveref}

\floatname{algoflt}{Protocol}
\floatstyle{ruled}
\newfloat{algoflt}{tbp}{loa}
\makeatletter
\lst@AddToHook{OnEmptyLine}{\addtocounter{lstnumber}{-1}\vspace{-0.5\baselineskip}}

\def\dcmnumberstyle{}
\lstnewenvironment{algorithm}[2][htb]{%
\lstset{%
    mathescape,%
    tabsize=2,%
    numbers=left,%
    numberstyle=\footnotesize\dcmnumberstyle,%
    numberblanklines=false,%
    numbersep=0.5em,%
    showlines=true,%
    basicstyle=\rmfamily\small,%
    backgroundcolor={},%
    columns=fullflexible,%
    emph={%
        [1]if, else, then, for, do, while, return, exit, output, initialize, var %
        },
    emphstyle={%
        [1]\bfseries%
        },%
    escapeinside={/*}{*/},
    escapebegin={\ifmmode\else\hfill\color{gray}\small$\triangleright$ \fi},
    escapeend={\ifmmode\else\fi},
    }%
\setlength{\topsep}{0pt}%
\setlength{\parindent}{0pt}%
\@skiphyperreftrue%
\algoflt[#1]\caption{#2}%
}{\endalgoflt\@skiphyperreffalse}
\makeatother

\def\Vardef#1{%
    \expandafter\newcommand\csname #1\endcsname[1]{%
        \def\first{##1}%
        \def\second{*}%
        \def\third{}%
        \ensuremath{\mathsf{\MakeLowercase #1}\ifx\first\second\else\ifx\first\third\else[{##1}]\fi\fi}%
    }%
}
\def\KWdef#1{%
    \expandafter\newcommand\csname #1\endcsname{%
    \mathsf{\MakeLowercase #1}%
    }%
}

\usepackage{xcolor}

\usepackage[textsize=scriptsize]{todonotes}
\usepackage{enumitem}

\makeatletter
\g@addto@macro\bfseries{\boldmath}
\makeatother

\newtheorem{theorem}{Theorem}[section]

\newtheorem{corollary}[theorem]{Corollary}
\newtheorem{lemma}[theorem]{Lemma}
\newtheorem{observation}{Observation}[section]

\newtheorem{definition}[theorem]{Definition}

\newcommand\bc[1]{\left({#1}\right)}
\newcommand\cbc[1]{\left\{{#1}\right\}}

\newcommand\uppergauss[1]{\left\lceil{#1}\right\rceil}

\newcommand\ceil[1]{\uppergauss{#1}}

\newcommand{\Z}{{\mathbb{Z}}}

\newcommand{\Oh}{{\operatorname{O}}}

\newcommand{\vX}{\vec{X}}

\renewcommand{\Pr}{\mathbb{P}}

\newcommand{\SmallProb}[1]{{\operatorname{\Pr}[#1]}}

\newcommand{\Exp}[1]{{\operatorname{\mathbb{E}}\mathopen{}\left[#1\right]\mathclose{}}}

\newcommand{\Prob}[1]{{\operatorname{\Pr}\mathopen{}\left[#1\right]\mathclose{}}}

\newcommand{\Bin}{{\textrm{Bin}}}

\newcommand{\cC}{{\mathcal{C}}}

\newcommand{\FuncName}[1]{\textsc{#1}}
\newcommand{\Ranking}{\ensuremath{\FuncName{AssignRanks}_r}}
\newcommand{\AssignRanks}{\Ranking}
\newcommand{\SheriffElection}{\FuncName{ElectSheriff}}
\newcommand{\Sleep}{\FuncName{Sleep}}
\newcommand{\Deputize}{\FuncName{Deputize}}
\newcommand{\Labelling}{\FuncName{Labeling}}
\newcommand{\LoadBalance}{\FuncName{BalanceLoad}}
\newcommand{\ErrorDetection}{\ensuremath{\FuncName{DetectCollision}_r}}
\newcommand{\DetectCollision}{\ErrorDetection}

\newcommand{\CheckMessageConsistency}{\FuncName{CheckMessageConsistency}}
\newcommand{\UpdateMessages}{\FuncName{UpdateMessages}}
\renewcommand{\Reset}{\FuncName{Reset}}
\newcommand{\Un}{\ensuremath{\FuncName{Unlabeled}_{t}}}
\newcommand{\ElectLeader}{\ensuremath{\textsc{ElectLeader}_r}}
\newcommand{\PropagateReset}{\FuncName{PropagateReset}}
\newcommand{\TriggerReset}{\FuncName{TriggerReset}}
\newcommand{\StableVerify}{\ensuremath{\FuncName{StableVerify}_r}}
\newcommand{\FastLeaderElect}{\ensuremath{\FuncName{FastLeaderElect}}}

\newcommand{\NoResetBy}[1]{\ensuremath{\mathsf{NoResetBy}_{#1}}}
\newcommand{\Interactions}[1]{\ensuremath{\mathsf{Ints}_{#1}}}

\newcommand{\FieldName}[1]{\textup{\texttt{#1}}}
\newcommand{\hrank}{\FieldName{highBadge}}
\newcommand{\lrank}{\FieldName{lowBadge}}
\newcommand{\id}{\FieldName{id}}
\newcommand{\counter}{\FieldName{counter}}
\newcommand{\lbl}{\FieldName{label}}
\newcommand{\sleep}{\FieldName{sleep}}

\newcommand{\channel}{\FieldName{channel}}
\newcommand{\messages}{\FieldName{msgs}}
\newcommand{\observations}{\FieldName{observations}}
\newcommand{\rank}{\FieldName{rank}}
\newcommand{\tagVal}{\FieldName{signature}}
\newcommand{\resetCount}{\FieldName{resetCount}}
\newcommand{\delayTimer}{\FieldName{delayTimer}}
\newcommand{\role}{\FieldName{role}}
\newcommand{\qRanking}{\FieldName{qAR}}
\newcommand{\qStableVerify}{\FieldName{qSV}}
\newcommand{\pState}{\qRanking}
\newcommand{\countdown}{\FieldName{countdown}}
\newcommand{\probationTimer}{\FieldName{probationTimer}}
\newcommand{\qDetectCollision}{\FieldName{qDC}}
\newcommand{\qVerifying}{\qDetectCollision}
\newcommand{\vState}{\qVerifying}
\newcommand{\generation}{\FieldName{generation}}
\newcommand{\Coin}{\FieldName{Coin}}
\newcommand{\Coins}{\FieldName{Coins}}
\newcommand{\CoinCount}{\FieldName{CoinCount}}
\newcommand{\LeaderBit}{\FieldName{LeaderBit}}
\newcommand{\LeaderDone}{\FieldName{LeaderDone}}
\newcommand{\Identifier}{\FieldName{Identifier}}
\newcommand{\MinID}{\FieldName{MinIdentifier}}
\newcommand{\LECount}{\FieldName{LECount}}

\newcommand{\RoleName}[1]{\textup{\textsf{#1}}}
\newcommand{\resetting}{\RoleName{Resetting}}
\newcommand{\producing}{\RoleName{Ranking}}
\newcommand{\ranking}{\producing}
\newcommand{\verifying}{\RoleName{Verifying}}

\newcommand{\cepi}{\ensuremath{c_{\textrm{epi}}}}

\newcommand{\csleep}{c_{\sleep}}
\newcommand{\cprobation}{c_{\text{prob}}}
\newcommand{\crank}{\ensuremath{c_{\textrm{ranking}}}}

\newcommand{\safe}{\ensuremath{\mathcal{C}_{\text{safe}}}}

\newcommand{\echo}{\textsf{EchoChamber}}
\newcommand{\effn}{{\nicefrac{n^2}{r} \cdot \log n}}
\newcommand{\effnovern}{{\nicefrac{n}{r} \cdot \log n}}
\newcommand{\QPropagateReset}{Q_{PR}}
\newcommand{\QRanking}{Q_{AR}}
\newcommand{\QAssignRanks}{\QRanking}
\newcommand{\QStableVerify}{Q_{SV}}
\newcommand{\qNaughtRanking}{{q_{0,AR}}}

\newcommand{\qNaughtStableVerify}{{q_{0,SV}}}
\newcommand{\soundness}{soundness}

\newcommand{\QDetectCollision}{Q_{DC}}
\newcommand{\QVerify}{\QDetectCollision}
\newcommand{\qNaughtDetectCollision}{{q_{0,DC}}}

\title{
  A Space-Time Trade-off for Fast Self-Stabilizing \\ 
  Leader Election in Population Protocols}
\date{}

\author[1]{Henry Austin}
\affil[1]{Durham University}
\author[2]{Petra Berenbrink}
\affil[2]{University of Hamburg}
\author[1]{Tom Friedetzky}
\author[2]{Thorsten Götte}
\author[2]{Lukas Hintze}

\usepackage[font=small,labelfont=bf,skip=2pt]{caption}
\usepackage{wrapfig}

\begin{document}

\maketitle
\thispagestyle{empty}

\begin{abstract}
\setstretch{1.2}

We consider the problem of self-stabilizing leader election in the population model by \citeauthor{DBLP:journals/dc/AngluinADFP06} (JDistComp '06). 
The population model is a well-established and powerful model for asynchronous, distributed computation with a large number of applications. 
For self-stabilizing leader election,
    the population of $n$ anonymous agents, interacting in uniformly random pairs, must stabilize with a single leader from \emph{any} possible initial configuration.

The focus of this paper is to develop time-efficient self-stabilizing protocols whilst minimizing the number of states. 
We present a parametrized protocol, which, for a suitable setting, achieves the asymptotically optimal time
$\Oh(\log n)$ using $2^{\Oh(n^2 \log n)}$ states (throughout the paper, ``time'' refers to ``parallel time'', i.e., the number of pairwise interactions divided by $n$).
This is a significant improvement over the previously best protocol $\textsc{Sublinear-Time-SSR}$ due to \citeauthor{DBLP:conf/podc/BurmanCCDNSX21} (PODC '21), which requires $2^{\Oh(n^{\log n} \log n)}$ states for the same time bound.
In general,  for $1\le r\le n/2$, our protocol requires $2^{O(r^2\log{n})}$ states and stabilizes in time $O((n\log{n})/r)$, w.h.p.; the above result is achieved for $r = \Theta(n)$.
For $r= \log^2 n$ our protocol requires only sub-linear time using only $2^{\Oh(\log^3 n)}$ states,
    resolving an open problem stated in that paper.
$\textsc{Sublinear-Time-SSR}$
    requires $O(\log(n) \cdot n^{1/(H+1)})$ time using
    $2^{\Theta(n^H) \cdot \log n}$ states for all
    $1\le H \le \Theta(\log n)$.

Similar to previous works, it solves leader election by assigning a unique rank from $1$ through $n$ to each agent.
The principal bottleneck  for self-stabilising ranking usually is to detect if there exist agents with the same rank. 
One of our main conceptual contributions is a novel technique for collision detection.

\setstretch{1.0}
\end{abstract}

\vspace{\fill}

{\footnotesize \noindent Petra Berenbrink and Thorsten Götte are supported by DFG grant no.\ 491453517. Petra Berenbrink and Lukas Hintze are supported by DFG grant no.\ 411362735. We thank the anonymous reviewers at PODC'25 for their valuable and in-depth feedback on an earlier version of this work.}

\newpage 

\section{Introduction} 

In this paper, we consider the problem of self-stabilizing leader election in the
population model introduced by \textcite{DBLP:journals/dc/AngluinADFP06}. 
The population model is a well-established and powerful  model for asynchronous, distributed computation with a large number of applications, ranging from rumour spreading in social networks \cite{DBLP:conf/podc/BerenbrinkCGMMR23}, chemical reactions \cite{DBLP:conf/soda/Doty14}, to gene regulatory networks \cite{DBLP:journals/ijon/SchilstraRAB02}.
In the model, the $n$ agents are (a priori) anonymous and have limited memory and computational capabilities.
Each agent is in one state out of a fixed set of states.
In each round, a pair of agents is chosen uniformly at random to interact;
    they update their states according to a fixed transition function.
The complexity of a population protocol is measured by counting the number of states that each agent can adopt and the number of pairwise interactions the protocol takes to achieve its goal. 
We often measure the latter in terms of \emph{time}, which is the average number of interactions of each agent, i.e., the number of interactions divided by $n$.

The objective of a leader election protocol is to mark \emph{exactly} one of
the $n$ anonymous agents as \emph{leader}. Leader election is perhaps one of
the most fundamental and important primitive operations for distributed
computation, and it is used to achieve coordination in many different settings.
In particular, population protocols with a leader are known to become much
faster \cite{DBLP:journals/dc/AngluinAE08a} and much more state-efficient
\cite{DBLP:conf/stacs/BlondinEJ18}. Thus, unsurprisingly, there is a plethora
of population protocols for leader election
\cite{DBLP:conf/icalp/AlistarhG15,
DBLP:conf/podc/BilkeCER17,
DBLP:conf/soda/AlistarhAG18,
DBLP:conf/soda/BerenbrinkKKO18,
DBLP:conf/soda/AlistarhAEGR17,
DBLP:conf/soda/GasieniecS18,
DBLP:journals/tpds/SudoOIKM20,
DBLP:conf/spaa/GasieniecSU19,
DBLP:conf/stoc/BerenbrinkGK20}. 
The most recent of these papers,
\cite{DBLP:conf/stoc/BerenbrinkGK20}, presents a protocol with $O(\log \log n)$
states that converges within $O(n\log n)$ interactions. This is
asymptotically optimal. Many of these protocols (including \cite{DBLP:conf/stoc/BerenbrinkGK20}) 
assume that all agents start in
some well-defined configuration. However, in large distributed systems, failures and
errors and the absence of well-defined initial configurations are the rule rather than the exception. 
In particular, the agents'
memory and, therefore, their states can be corrupted through all kinds of
outside influences. For this reason, we need to design and investigate protocols that not
only elect a leader efficiently but also resiliently and reliably. Here, the notion of
self-stabilization comes into play. Self-stabilization as introduced by
\textcite{DBLP:journals/cacm/Dijkstra74} is a powerful paradigm for reliable
protocols. The main feature of self-stabilization in population protocols is for them to reach a
\emph{correct} configuration  with probability one, starting from an
\emph{arbitrary} (possibly invalid, nonsensical) configuration. In our case, the set of correct configurations
comprises all those where exactly one agent is marked as a leader.

As our main result, we present the protocol $\ElectLeader$ that is both fast and self-stabilizing. It solves leader election with  $(\frac{n^2}{r} \log{n})$ interactions w.h.p.\footnote{We say that an event occurs \emph{with high probability} (or \emph{w.h.p.})\ if it occurs with probability at least $1 - \Oh(n^{-1})$.}\ using $2^{O(r^2\log{n})}$ states for all $r$ with $1\le r\le n/2$. 
With $r=\Theta(n)$, our protocol reaches the optimal number of $O(n\log n)$ interactions using $2^{\Oh(n^2 \log n)}$ states.
This is a significant improvement of the previous state-of-the-art, the protocol by \textcite{DBLP:conf/podc/BurmanCCDNSX21}, which requires $2^{n^{O(\log n)}}$ states for the same (asymptotic) number of interactions.
On the other end of the spectrum, for $r = \log^2(n)$ our protocol  requires only a sub-quadratic number of interactions using a sub-exponential number of $O(2^{\log^3 n})$  states, solving an open problem posed in \cite{DBLP:conf/podc/BurmanCCDNSX21}. 
For constant $r$ our bound matches the bound of \cite{DBLP:conf/podc/BurmanCCDNSX21} for silent protocols.

\subsection{Model and our Contribution}

Throughout this paper, we write $[i]=\{1,...,i\}$, and all logarithms are taken to be natural unless explicitly stated otherwise.
As is standard in the population protocol model, we consider \emph{populations} of $n$ indistinguishable \emph{agents}, each of which takes a state from a predefined state space $Q$. 
All agents are equipped with a fixed \emph{transition function} $\delta: Q\times Q\to Q \times Q$, such that when two agents interact, their states are updated according to $\delta$.
The pair $Q$ and $\delta$ define a protocol using $|Q|$ states.
The protocols we present in this work are \emph{strongly non-uniform}, meaning that $n$ is encoded in the transition function, and the protocol is different for every $n$.
This is, in fact, necessary for self-stabilizing leader election as shown by \textcite{DBLP:journals/mst/CaiIW12}.

A \emph{configuration} $C \in Q^n$ is a vector specifying the state of all agents in the population.
We assume a \emph{uniformly random scheduler}: in each step, a uniformly random pair of agents interacts, updating their states according to $\delta$.
A configuration $C'$ is \emph{reachable} from $C$ if there exists a sequence of pairs of agents such that by application of $\delta$ to the states of each pair in turn, $C$ is transformed into $C'$.
We call a  configuration  \emph{stable} for a given problem if it, and all configurations reachable from it, are correct (in our context, exactly one agent is marked as a leader). 
A population protocol with state space $Q$ is \emph{self-stabilizing} with respect to a set of configurations $C_L \subset Q^n$ if and only if it fulfills the following two properties:
\begin{itemize}[noitemsep,topsep=0pt]
    \item \emph{Closure:}\quad If $\vX_t \in C_L$ for some $t$, then $\vX_{t+1} \in C_L$.
    If additionally $\vX_{t+1} = \vX_{t}$, i.e., no agent changes its state, the protocol is \emph{silent}.
    \item \emph{Probabilistic Stabilization:}\quad For every $\vX_t \in Q^n$ we have $\lim_{\tau \to \infty} \SmallProb{\vX_{t+\tau} \in C_L}=1$.
\end{itemize}
Note that in contrast to other models, we cannot guarantee deterministic stabilization for population protocols due to the random interactions.

Recall that transition function $\delta$ is deterministic and cannot generate random numbers.
Nevertheless, for an easier presentation of our protocols, we assume that the agents have access to random values.
That is, in each interaction, we assume that an agent can sample a value $X \in [N]$ (almost) u.a.r.\ where $N$ is some integer,
    where by \emph{almost} u.a.r., we mean that $\Pr{[X = x]} \in [\frac{1}{2N},\frac{2}{N}]$ for all $x \in [N]$.
Fortunately, the required randomness can be simulated through known techniques 
that exploit the randomness of the scheduler without bloating the state space too much. 
We explain how our protocols can be derandomized in \autoref{apx:derandomization}.

\medskip

\noindent Our algorithm  $\ElectLeader$ implements a time-space trade-off and matches or improves upon the current state of the art in all regimes using at most $o(n^2)$ interactions.

\begin{theorem}[Main Theorem]
\label{thm:framework}
Assume $1 \leq r < n/2$.
    The protocol $\ElectLeader$  solves self-stabilizing leader election and ranking using $2^{\Oh(r^2 \log n)}$ states and  $\Oh(\nicefrac{n^2}{r} \log n)$ interactions w.h.p.
\end{theorem}

For $r = \Theta(n)$, our protocol solves self-stabilizing ranking in the optimal $\Oh(n\log{n})$ interactions w.h.p.\ using only $2^{\Oh(n^2\log{n})}$ states as opposed to the previous best of $2^{\Oh(n^{\log{n}})\log{n}}$ states in the same 
regime~\cite{DBLP:conf/podc/BurmanCCDNSX21}.
To emphasize, we reduce the bit complexity (i.e., the logarithm of the size of the state space) of time-optimal self-stabilizing leader election from super-polynomial to sub-cubic.
Furthermore, we extend the range of the time-space trade-off to include the regime with sub-exponential states.
This is summarized in the following result.

Our result gives an affirmatively answer an open question of \cite{DBLP:conf/podc/BurmanCCDNSX21} on the existence of a protocol solving self stabilizing leader election in sub-linear time with only a sub-exponential number of states.
With $r=\Theta(\log^2(n))$, our protocol solves self-stabilizing leader election in $\nicefrac{n^2}{\log{n}}$ interactions w.h.p.\ using only $2^{\log(n)^3}$ states.
In fact, our collision detection protocol on its own constitutes a positive answer to the weaker question also posed in \cite{DBLP:conf/podc/BurmanCCDNSX21} on the existence of an ``initialized collision detection'' protocol in the same regime, independently of its recent resolution in \cite{araya2024sublinear}.
Our overall trade-off requires the construction of bespoke ranking protocols with space-time trade-offs, which we include for completeness.

\subsection{Structure of the Paper}

The remainder of this paper is structured as follows:
In \autoref{sec:relatedwork} we present some related work and important previous results.
Then in \autoref{sec:highlevel}, we provide a non-technical high-level overview of our protocol and its main ideas.
 We then describe the main protocol and its sub-protocols in \cref{sec:overall_protocol,sec:verify} (as well as \autoref{apx:ranking}, due to fierce competition for space among the first ten pages of the paper).
For a hopefully more intuitive presentation and a better reading flow, we will omit many if not most technicalities from these descriptions --- those are instead presented across the various appendices. 
Finally, we provide the framework of the proof of \cref{thm:framework} in \cref{sec:analysis}.

\section{Related Work}
\label{sec:relatedwork}

In the following, we present an overview of related work in similar models of computation or with similar notions of stabilization.
We group these works into five categories.

\paragraph{Self-Stabilizing Leader Election} 
Recall that in this work, we aim to optimize the number of states that are required for protocols that solve self-stabilizing leader election in a given near-optimal time.
In this regime, to the best of our knowledge, only the aforementioned work of \textcite{DBLP:conf/podc/BurmanCCDNSX21} present a comparable protocol. 
However, other works on self-stabilizing leader election  nicely complement our work by  aiming to optimize the stabilisation time while using a given near-optimal number of states.
Notably, \textcite{DBLP:journals/mst/CaiIW12} present a self-stabilizing leader election protocol using only $n$ states and time $O(n^2)$ in expectation. 
\cite{DBLP:conf/podc/BurmanCCDNSX21} improves the stabilization time  from $O(n^2)$ to $O(n)$ in expectation and $O(n\log n)$ w.h.p. Their protocols require slightly more, namely $\Theta(n)$, states. 
Very recent work by \textcite{leszekArXiv} complement our work by optimizing  the state complexity even further.
They propose (among other results) an algorithm that uses only $n + O(\log n)$ states and stabilizes in time $O(n\log n)$ w.h.p.
Furthermore, with $n+1$ states (i.e., one additional state compared to \cite{DBLP:journals/mst/CaiIW12}) their protocol still stabilizes in time $O(n^{\nicefrac{7}{4}} \log^2 n)$.
Notably, all these protocols \cite{DBLP:journals/mst/CaiIW12,DBLP:conf/podc/BurmanCCDNSX21,leszekArXiv} are \emph{silent}, and all solve the problem via ranking, i.e., they assign a unique rank from $[n]$ to each agent.
In the silent setting, it is known that leader election requires at least $n$ states and time $\Omega(n)$ on expectation and $\Omega(n \log n)$ for a high probability result \cite{DBLP:conf/podc/BurmanCCDNSX21}.
This makes the algorithms of \cite{DBLP:conf/podc/BurmanCCDNSX21} and \cite{leszekArXiv} time-optimal and (almost) state-optimal in the silent regime.

\paragraph{Non Self-Stabilizing Leader Election}
There is a large number of results for non self-stabilizing leader election in the population protocol model where it is assumed that all agents are initially in the same state, usually something like ``potential leader'' \cite{DBLP:conf/icalp/AlistarhG15,DBLP:conf/podc/BilkeCER17,DBLP:conf/soda/AlistarhAG18,DBLP:conf/soda/BerenbrinkKKO18,DBLP:conf/soda/AlistarhAEGR17,DBLP:conf/soda/GasieniecS18,DBLP:journals/tpds/SudoOIKM20,DBLP:conf/spaa/GasieniecSU19,DBLP:conf/stoc/BerenbrinkGK20}. 
The most recent of these papers, \cite{DBLP:conf/stoc/BerenbrinkGK20}, presents a protocol  which uses $O(\log \log n)$ states and converges in optimal time $O(\log n)$ on expectation and time $O(\log^2 n)$  w.h.p.
Moreover, \textcite{burman_et_al:LIPIcs.DISC.2019.9} present a comprehensive study of the necessary and sufficient state space conditions for non-self-stabilizing ranking.

\paragraph{Non Self-Stabilizing Ranking}
 \textcite{DBLP:conf/opodis/GasieniecJLL21,gasieniec2021efficient} 
focus on  safe and silent ranking protocols. 
A ranking protocol is called \emph{safe} if no agent ever adopts more than one rank during the execution of the protocol, that is, the first rank assigned to an agent is final. 
Their first protocol assigns a rank  in $1,\ldots,(1+\epsilon) \cdot n$ requiring $O(n/\epsilon\log n)$ interactions w.h.p. The protocol uses $(2+\epsilon) n+O(n^{\alpha})$ states for an arbitrary constant $\alpha<1$.
For the optimal range of $[1, n]$ the authors present a protocol which needs time $O(n^2)$ in expectation and $O(n+5\sqrt{n}+O(n^c))$ states, where $c$ can be an arbitrarily small constant. They also present a parameterized version of the latter protocol.

\paragraph{Loosely Self-stabilizing Leader Election}
An interesting relaxation of self-stabilization are so-called  
\emph{loose-stabilization} guarantees introduced in \cite{DBLP:journals/tcs/SudoNYOKM12}, where it is required that the population
reaches a safe configuration within a relatively short time, starting from any configuration; after that, the unique leader must be sustained for a certain time referred to as \emph{holding time}, but not necessarily, or even in general, forever. 
\textcite{DBLP:conf/wdag/SudoEIM21} present a loosely-stabilizing leader election protocol with parameter $\tau\ge 1$. It selects a leader in expected time $O(\log n)$ and expected holding time $O(n^{\tau})$ using $O(\tau \log n)$ states.
This improves on earlier work on leader election in this model
\cite{DBLP:journals/tcs/SudoNYOKM12, DBLP:journals/tcs/SudoOKMDL20}.

\paragraph{Leader Election in Anonymous Networks}
Another related problem is that of assigning a rank to all nodes of an anonymous dynamic network
    modelled as a connected graph $G = (V, E)$ whose edges may change over time.
Time proceeds in synchronous rounds and in each round, a node $v \in V$ can send a message to all its neighbours.
The nodes have no identifiers but nodes can differentiate between their neighbours based on \emph{port numbers}.
\Textcite{DBLP:conf/icdcs/KowalskiM21} present a non-self-stabilizing leader election protocol with a running time that is a function on mixing time of a simple random walk on $G$.
\Textcite{DBLP:conf/focs/LunaV22, DBLP:conf/mfcs/LunaV24} consider the ``reverse'' problem of determining the number of nodes $n$ given a predetermined leader.
Note that algorithms for population protocols can usually be transferred to anonymous networks
\textcite{DBLP:conf/opodis/AlistarhGR21} with runtime that depends on graph properties (diameter or conductance).

\section{Overview of our Algorithm}
\label{sec:highlevel}

In this section, we describe the main ideas of our protocol $\ElectLeader$ and its subprotocols. 
More details (including pseudocode) can be found in \autoref{sec:overall_protocol} and \autoref{sec:verify}.
The fundamental difficulty of self-stabilizing leader election in the population model is detecting \emph{either} the absence of a leader \emph{or} the existence of duplicate leaders.
To this end, we adopt the de facto standard technique of de-anonymising the population by assigning a unique rank from $[n]$ to each agent and taking the agent with rank $1$ to be the leader.
Under this framing, \emph{both} the existence of multiple leaders \emph{and} the absence of a leader require two agents to have the same but supposedly unique rank. 
However, the process of detecting such a collision of ranks is the principal bottleneck in the previous state-of-the-art for fast self-stabilizing leader election~\cite{DBLP:conf/podc/BurmanCCDNSX21}.

At a very high level, our algorithm consists of three components: \emph{ranking}, \emph{collision detection}, and \emph{resetting}.
First, the ranking component attempts to assign a rank to all agents. 
If the population starts from a \emph{clean} configuration where no agent has a rank, the ranking component assigns a unique rank to each agent w.h.p. 
Once every agent has a rank, the collision detection component continuously checks for the existence of two agents with the same rank.
Every agent maintains a timer to ensure that all agents eventually transition to the collision-detection component, even if the ranking fails to rank the agents consistently.
In a legal configuration, no such collision is ever detected.
However, it is possible that either (a) the ranking component failed or (b) we started from an inconsistent configuration (e.g., a configuration with two leaders or an incorrectly ranked initial configuration).
Should such a collision ever be detected, the resetting component is triggered and all agents' states are reset to a clean initial configuration and the ranking component is started again.
 
While this overall approach has been employed before, we make three main contributions, which are described below in more detail. The first contribution is a novel message-based collision detection mechanism. 
This mechanism efficiently detects collisions, but it might also need to be reset after the protocol has already assigned unique ranks. 
To mitigate this, our second contribution is a technique for safely resetting only the collision-detection component after stabilisation has occurred. 
The third contribution is a space-time trade-off adoption of our approach.

\subsection{Collision Detection}
\label{sec:overview:collision}

Our first major technical contribution is a novel technique for collision detection (see \autoref{sec:error_detection}), which allows us to dispense with the costly history tree method of~\cite{DBLP:conf/podc/BurmanCCDNSX21}, which was the state-size bottleneck in their algorithm.
Recall that the collision detection runs after the ranking has concluded, so that we can assume that each agent has a rank (though it may not be unique).
The core difficulty in collision detection for SSLE is that the system must not generate false positives (i.e., report a collision when none exists), but must be fast and correct from any configuration.
We, therefore, require some form of proof of a collision.
The simplest such proof would be if two agents of the same rank interact, however, this will typically take $\Omega(n)$ time.
We resolve this by amplifying the number of objects between which we can detect collisions.
Loosely speaking, the agents pass around a large number of identifiable tokens which we call \emph{messages}, each of which contains the rank of the generating agent and a random string.
The message can only be modified by agents having the same rank as the one in the message. Furthermore, agents always keeps a copy of messages they generated. 
Whenever an agent encounters one of its own messages (i.e., one it has permission to modify), it modifies the message's string and records its new value.
If there is only a single agent of each rank, the string within each message will always exactly match that of the corresponding agent.
However, should there be multiple agents of the same rank, eventually one of them will modify a message which then conflicts with the record of the other agents.
The moment an agent sees such a conflict it can safely conclude that there exists another agent with the same rank.

In the following, we present more details of our messaging approach. 
As soon as the collision detection is started (which is after the ranking component has just concluded),
    every agent generates $M=\Theta(n^2)$ messages.\footnote{For technical reasons, the initial round of messages for each rank is actually hardcoded into the protocol's transition function, and messages are ``pre-mixed'' among agents.}
The messages are triples of the form $(\text{rank}, \text{ID}, \text{content})$.
The agents  populate the $\text{rank}$ entry with  their own rank and the $i$th of these messages receives $i$ as the $\text{ID}$ entry of the message. 
Hence, in a consistent configuration (where the ranks are distinct), all messages have a unique $(\text{rank}, \text{ID})$ pair.
The $\text{content}$ entry of the message is filled with a randomly generated string.
Furthermore, the agents keep local copies of these messages.
When an agent with rank $i$ encounters another agent that has a message whose rank component is $i$,
    it checks the message's content against its local copy (using the message's ID to identify which of its generated message it dealing with).
If there is no conflict,  it refreshes the content of the message, updating its local copy in turn (the old content is discarded).

There are two more technical details we have to explain.
The first one is the way that we generate the random message contents, and the second one is how agents exchange these messages during interactions.
Recall that the only source of randomness for a population protocol comes from interactions.
Roughly speaking, using the standard technique, an agent can generate one random bit with each interaction.
Hence, it is important to  limit the required randomness for generating message content.
Using new random bits every time a message is updated makes it impractical.
Instead, each agent maintains a \emph{signature} drawn from $[\Theta(n^5)]$ which requires $O(\log n)$ random bits (see \autoref{apx:derandomization}).
Whenever an agent needs a random string to use as message content, it simply uses its signature.
Agents resample the signature every $\Theta(\log{n})$ interactions in case two agents of the same rank were initialised with the same signature.
We implement this by storing, in addition to the signature in use, a partial ``stand-by'' signature which is extended by one random bit every interaction, and which replaces the signature in use when it is completed.
Now, for exchanging messages, we require a mechanism that ensures that refreshed messages are distributed quickly among the agents, so that agents have a sufficiently high probability of encountering a conflicting message.
Without such a mechanism, the messages would stay clumped together.
Hence, whenever two agents interact, they swap messages according to the following scheme: 
Consider all messages with matching rank and content independently, i.e., all messages that are (allegedly) assigned the same signature by the same agent.
For every such (rank, content) pair, the corresponding messages are swapped until both agents either hold the same number of these messages or the number differs by at most one.
Note that this can be achieved deterministically based on the IDs, so no random bits are required.
All messages are ordered by their ID, and one agent receives the first half and the other receives the second half.
For technical reasons, we must enforce that all agents hold $\tfrac{M}{n} \in \Theta(n)$ messages of each rank at all times.
Note that the mechanism sketched above maintains this invariant.

With the technical details clarified, we can now give a more detailed account now how the collision detection works.
Suppose every agent holds $\Theta(n)$ messages belonging to each rank at all times.
Thus, an agent updates $\Theta(n)$ messages to match its current signature whenever it is activated.
Our load-balancing mechanism then ensures that after $O(n\log{n})$ steps, all agents hold at least one of these $\Theta(n)$ newly updated messages w.h.p.
Since a signature only gets updated every $\Omega(n\log n)$ interactions, w.h.p., the updated messages can spread without their content being overwritten.
Now, assume there are two agents with the same rank. 
Within $O(n\log{n})$ steps, one of the two will adopt a signature distinct from both the other agent and the content of all messages beginning with their shared rank w.h.p.
This agent will update $\Theta(n)$ messages, such that they no longer match the stored local copies held by its competitor.
The competitor will encounter an updated message and, when it does not match its stored local copy, trigger a full reset within $O(n\log{n})$ rounds w.h.p.

Summing up, the mechanism sketched above will detect a rank collision within $O(n \log n)$ steps w.h.p.\ if it is properly initialized with $\Theta(n)$ messages of each rank at each agent and the messages' contents matching the values stored by the respective agents.
However, given that we may start from any initial configuration, we must also deal with errors in this component.
In particular, if no agents with duplicate rank exist, we must deal with the case
where 
the message system contains invalid messages despite the valid ranking. For example, when
all but one messages are valid, it typically takes $\Omega(n^2\log n)$
interactions for the single invalid message to reach the agent having the
message's rank. Only then will the agent trigger a reset, which is
potentially disastrous because a full reset of the protocol will destroy the
correct ranking. To address this, we introduce a \emph{soft} reset mechanism that is described in the next section.

\subsection{Soft Reset Mechanism}
\label{sec:overview:softreset}
In this section, we describe how to handle errors in the collision detection component described in the previous section.
First off, if an error is detected early in the algorithm's execution, say, within $O(n\log n)$ steps, we can safely perform a full reset. 
However, if it is detected quite late, we have to avoid a full reset unless there are indeed duplicate ranks, as doing so will destroy the existing ranking.
Fortunately, our algorithm is able to detect rank collision quite fast. Therefore, any error detected after a long time is likely to be caused by an error in the messaging system (which can occur due to an unfortunate initial state).
To exploit this, we introduce two interleaved mechanisms: the \emph{soft reset} mechanism, which resets only the message system, and the  \emph{probation} mechanism to decide whether to perform a soft or hard reset. We consider these to be our second major technical contribution (see \cref{sec:selfstab}) of this paper. It should be noted that our construction
applies ranking and collision detection as sub-protocols in a black
box manner. Hence, the approach can be applied to other problems.
    
First we describe our \emph{probation mechanism} which is implemented via a timer (counting down from an suitably chosen value) maintained at each agent.
The probation timer is reset whenever the agent either directly encounters inconsistent messages or is informed about a soft-reset as a consequence of another agent finding inconsistent messages.
Upon finding inconsistent messages, the agent checks its probation timer.
\begin{itemize}\itemsep0pt
\item If the probation timer is zero, w.h.p. a large amount of time has passed since any inconsistencies were detected.
Since our collision detection mechanism will detect a real collision quickly w.h.p., the agent can conclude with a high degree of confidence that any inconsistency between messages was likely the result of a bad initialisation of the collision detection system, as opposed to a real collision between ranks.
Consequently, the agent initiates a soft-reset which propagates through the population (see below).
\item If the probation timer is positive, only a short period of time has passed since either the beginning of the process or the detection of the last inconsistency.
In the former case it is safe to reset the entire protocol as doing so will not seriously impact the stabilisation time.
In the latter case some inconsistency survived the last soft reset which occurs with only low probability, unless there exists a genuine collision between agents.
Therefore, to be safe the agent calls for a hard reset.
\end{itemize}

Note that our soft reset has the following property. 
After a successful soft reset, inconsistent messages exist if and only if the ranking is incorrect.
Thus, if the ranking is correct after a successful soft reset no further inconsistencies will be encountered in future and the correct ranking will be maintained forever with certainty.
On the other hand, if the ranking is incorrect inconsistent messages will still exist and will be discovered again before the agents' probation timers run out w.h.p.

\medskip

It remains to describe our soft reset approach in more detail. If an agent
detects an inconsistent message and the probation timer is zero, then it triggers a
soft reset. First it discards all its messages for circulation
as well as its own local message copies.
Then it re-initializes its $\Oh(n^2)$ collision detection messages as though first entering the collision detection protocol. 
(Note that due to the specifics of the model, all this is done during one interaction).
This propagates as described in the next paragraph.

We require that a soft reset can be safely
executed even if the protocol has already stabilized with a correct ranking
(i.e., the soft reset will not change a correct ranking). If there are no
duplicate ranks but there exist inconsistent messages, the soft reset will be
triggered within $\Oh(n^2 \log n)$ interactions, the contradiction in the
message system will be repaired and the ranking will remain invariant. However,
we need to avoid agents who have not yet undergone the soft reset on their own
state space putting old (and potentially inconsistent) messages into
circulation. To this end we equip each agent with a \emph{generation} variable 
which can take on values
in $\Z_6$ (the integers modulo~6).
Agents which have the same generation interact as usual.
An agent wanting to trigger a soft reset increments its generation, discards all messages,
and re-initializes its $\Oh(n^2)$ collision detection messages.
When an agent interacts with another agent whose generation is
larger by one (modulo~$6$), it adopts the successor generation and triggers a
soft reset of its own messages.
This way, the successor generation will be spread via broadcast.
As we will see later, we are able to show that, beginning
from a correct ranking, only a single soft reset will occur, and we are further
able to force the agents to quickly agree on a single generation. Hence,
counting the generations modulo $6$ is sufficient.

\newcommand{\ColDet}{\ensuremath{\FuncName{ColDet}}}
\subsection{Space-Time Trade-off}
\label{sec:overview:tradeoff}
Here we describe our generic technique for obtaining a trade-off between space and time for population protocols detecting collisions in rankings. Our method can be regarded as a framework for a wide class of collision detection protocols.
The trade-off is controlled by a parameter $1\leq r \leq \frac{n}{2}$, where bigger values of $r$ lead to faster protocols taking more memory.
Assume we are given any black-box collision detection protocol $\ColDet$ defined on all populations of size $\leq \nicefrac{n}{2}$. We partition the rank space into $\lceil\frac{n}{r}\rceil$ \emph{groups} of size $\Theta(r)$ and implement $\ColDet$ in each group independently,
    ignoring interactions between agents belonging to different groups.
Recall that a collision means two agents having the same rank. Hence, collisions can be discovered by interactions only within those agents' group and each group can be treated as a distinct (sub)population.
This essentially allows us to use $\ColDet$ parametrised for a population of size $r$,
    which allows agents to use significantly fewer states than for a population of size $n$.
However, it comes at the cost of an $\tilde{\Oh}(\nicefrac{n}{r^2})$ multiplicative slow-down due to the fact that only interactions between agents in the same group are meaningful for that group.
However,  the slow-down is somewhat ameliorated since the absolute number of interactions to detect a collision among $r$ agents is less than that for $n$ agents.

To prevent the ranking component from becoming a bottleneck, and in order to be able to use the same ranking protocol throughout the parameter space,
    we design a new ranking component.
At a high level, the protocol first elects a \emph{sheriff}, which in turn selects $r$ \emph{deputies}. Each deputy has a pool of $\lceil\nicefrac{cn}{r} \rceil$ unique labels, where $c>1$ is a constant.
Labels are distributed consecutively and all deputies continuously broadcast the number of labels they have assigned.
Since the total number of labels is larger than $n$ by a constant factor a constant fraction of the deputies will have unused labels in their pool. This allows the component to assign labels to the last unlabelled agents quickly.
Once agents learn that $n$ labels have been assigned (due to constantly broadcasting the number of assigned labels), they can choose a rank. They do that  according to some pre-determined mapping depending on the number of labels assigned by each of the deputies and their own label.

Note that the  space-time trade-off is obtained as a larger number of deputies increases the rate of label assignment, but, in turn, each agent must store the number of labels assigned by each deputy.
We defer details to \autoref{apx:ranking}.

\section{Description of \texorpdfstring{$\ElectLeader$}{Our Main Protocol}}
\label{sec:overall_protocol}

In this section we describe our protocol $\ElectLeader$ (Protocol~\ref{pr:wrapper}) in more detail.
The protocol is parametrized by $1\leq r\leq \nicefrac{n}{2}$, as introduced in \cref{sec:overview:tradeoff}, where bigger values lead to a faster protocol with a (much) larger state space.
As we will see, the protocol is a thin wrapper around a few submodules which we present in the next two sections.

\paragraph{State Space}

\begin{wrapfigure}{R}{.63\textwidth}
\[
\arraycolsep=0.5pt
\begin{array}{@{}ccccccccccccccccccc@{}}
  & &
  \graybox{3.5em}{\mathstrut \footnotesize\resetting} & &
  \graybox{8.5em}{\mathstrut \footnotesize\ranking} 
  & &
   \graybox{5em}{\mathstrut\footnotesize\verifying}\\
\underbrace{\left\{\substack{\resetting, \\ \ranking, \\ \verifying}\right\}}_{\texttt{role}} & \times  \Big(  & 
\underbrace{Q_{PR}}_{\texttt{qPR}}
& \uplus &
\underbrace{[\Theta(\nicefrac{n}{r} \log n)]}_{\texttt{countdown}}\times
\underbrace{Q_{AR}}_{\texttt{qAR}} & \uplus &
\underbrace{[n]}_{\texttt{rank}}\times\underbrace{Q_{SV}}_{\texttt{qSR}} & \Big) 
\\
\end{array}
\]
    \caption{An overview of $\ElectLeader$'s state space $Q$. The subspaces $Q_{PR}, Q_{AR}$, and $Q_{SV}$ and described in \cref{apx:PropagateReset}, \cref{apx:ranking}, and \cref{sec:verify} repectively.}
    \label{fig:state-space}
\end{wrapfigure}

We describe the agents' states using \emph{fields} notated using a monospace font.
In pseudocode, we refer to agent $i$'s field as $i.\FieldName{field}$.
An overview of $\ElectLeader$'s state space is given in \cref{fig:state-space}.
Each agent has a field $\role$ which can take one the three values $\resetting$, $\producing$, or $\verifying$; we call the agents in those roles \emph{resetters}, \emph{rankers}, and \emph{verifiers}.
Resetters execute the protocol $\PropagateReset$ described in \cref{apx:PropagateReset}. 
This protocol \emph{resets} the population to a well-defined \emph{clean} configuration.
Rankers execute $\Ranking$ described in \cref{apx:ranking} that assigns a unique rank from $[n]$ to each each agent (starting from a clean configuration).
Finally, Verifiers execute $\StableVerify$ described in \cref{sec:verify} that checks if the ranking is correct and resets the population if necessary.
We denote $\QPropagateReset$, $\QRanking$, and $\QStableVerify$ to be the local state (sub-)spaces of $\PropagateReset$, $\Ranking$, and $\ErrorDetection$ as described in the respective sections. 
In addition to these protocol specific state spaces, the agent have a field $\rank \in [n]$ that stores its presumed rank and field $\countdown \in [C_{max}]$ for a sufficiently large $C_{max} = \Theta(\frac{n}{r} \log n)$.
Further, we note that $\qNaughtStableVerify \in \QStableVerify$ is the local state with which $\StableVerify$ is initialized. 

Depending on its $\role$, only certain fields in the state space are used during interactions; these fields are called active.
Whenever an agent changes its role, all newly inactive fields are deleted at the end of the interaction.
For each role, the state space is (in general) obtained as the cross-product of the active fields' sets;
the total state space is obtained by the disjoint union of these cross-products (so that its size is the sum of the sizes of the various roles' state spaces).

\newcommand{\QDormant}{Q_{\mathrm{dormant}}}

\begin{algorithm}[!ht]{$\ElectLeader(u,v)$.\label{pr:wrapper}\label{pr:electleader}}
if $u.\role = \resetting$ then
  $\PropagateReset(u, v)$ /* non-resetters may become resetters, and resetters may become rankers */

if $u.\role = v.\role = \producing$ then
  $(u.\pState, v.\pState) \gets \Ranking(u.\pState, v.\pState)$ /* execute $\Ranking$ */
  $u.\countdown \gets u.\countdown - 1$; $v.\countdown \gets v.\countdown - 1$
  
for $(i,j) \in \cbc{(u,v),(v,u)}$ with $i.\role = \producing$ and ($i.\countdown = 0$ or $j.\role = \verifying$) do
  $i.\role \gets \verifying$ /* transform ranker into verifier when timer runs out, or via epidemic */
  $i.\rank \gets i.\qRanking.\rank$; $i.\qStableVerify \gets \qNaughtStableVerify$

if $u.\role = v.\role = \verifying$ then
  $\textsc{StableVerify}(u, v)$ /* detect collisions, trigger soft/full reset, perhaps making some agents resetters */
\end{algorithm}

\paragraph{Algorithm}
We give the pseudocode of $\ElectLeader$ as Protocol~\ref{pr:wrapper}.
Agents execute one of the submodules depending on their role,
    and store that submodules' local state in a corresponding field:
    Resetters execute $\PropagateReset$,
        rankers execute $\Ranking$,
        and verifiers execute $\StableVerify$.
Rankers additionally use the $\countdown$ field
    to handle the transition between $\Ranking$ and $\StableVerify$
    in a way that prevents $\Ranking$ from stalling indefinitely:
Whenever two rankers interact, they decrement their $\countdown$ by one, and if an agent has a $\countdown$ of $0$ they are forced to change their role from $\ranking$ to $\verifying$.
The rank computed by $\Ranking$ is stored in a subfield $\rank$ of $\QAssignRanks$;
    when agents change their role from $\ranking$ to $\verifying$, it is copied over into the global $\rank$ variable.
\newcommand{\qReset}{\FieldName{qPR}}

The protocol $\PropagateReset$ was presented in \cite{DBLP:conf/podc/BurmanCCDNSX21}. For reasons of completeness we 
  state the protocol in 
\cref{apx:PropagateReset} but mostly use it as a black-box.
We require only the following:
An agent wanting to reset the entire population \emph{triggers} a reset,
    with the resulting configuration being a \emph{triggered} configuration.
Agents which are not resetting (i.e., rankers and verifiers) are \emph{computing}.
Further, we say an agent is \emph{dormant} if it has successfully reset its state to some predefined initial state and is waiting to restart (which will happen within $\Oh(n \log n)$ interactions w.h.p.).
A configuration is \emph{fully dormant} if all its agents are dormant.
When the first agent starts computing again from a dormant configuration, we call this an \emph{awkening} configuration.
The protocol ensures that from a triggered configuration, we quickly reach a fully dormant and then an awakening configuration.

Due to space constraints, we present the $\Ranking$ and the formal statement of its properties in \cref{apx:ranking}.
In short we require that it is silent, uses $2^{\Oh(r^2\log{n}})$ states and assigns unique ranks within $\Oh(\nicefrac{n^2}{r}\log{n})$ interactions of a fully dormant configuration rank w.h.p.
The details for $\StableVerify$ are given in the next section.

\section{The Module \texorpdfstring{$\StableVerify$}{StableVerify}}
\label{sec:verify}
\label{sec:selfstab}

In this section, we describe the module $\StableVerify$, which verifies that the ranking of the agents is correct and handles the full and soft resets. 
The protocol consists two parts, a submodule $\DetectCollision$ (Protocol~\ref{pr:error}), which handles the collision detection itself and a wrapper protocol that decides whether to perform a full or a soft reset (Protocol~\ref{pr:framework}).
Here, we focus on the technical description, additionally filling in the details we left out in the non-technical description in \cref{sec:highlevel}.
This section is structured as follows: We first describe the wrapper protocol (using $\DetectCollision$ as a black-box) in the remainder of this section and then $\DetectCollision$ in \cref{sec:error_detection}.

\begin{wrapfigure}{R}{.45\textwidth}
\centering \[\arraycolsep=0.5pt
\begin{array}{@{}ccccccccccccccccccc@{}}
{\graybox{10em}{\mathstrut\textsc{\footnotesize Wrapper Fields}} }
 & &
  {\graybox{6.25em}{\mathstrut{\footnotesize\DetectCollision}} }\\
 \underbrace{\Z_6}_{\texttt{generation}} \times   

\underbrace{[\Theta(\nicefrac{n}{r} \log n)]}_{\probationTimer}
& \times &
\underbrace{\left[2^{\Oh(r^2\log{r})}\right]}_{\texttt{qDC}} 
\end{array}
\]
    \caption{An overview of $\StableVerify$'s state space. The field $\qDetectCollision \in \QDetectCollision$ is a state from $\DetectCollision$ described in \cref{sec:error_detection}.}
    \label{fig:state-space-stable-verify}
\end{wrapfigure}

\paragraph{State Space}
The local state space $\QStableVerify$ of $\StableVerify$ consists of two parts, the state space of the subprotocol $\DetectCollision$ and two additional fields for the generation number and the probation timer; see \cref{fig:state-space-stable-verify}.
Furthermore, the agents have access to the field $\rank$ from the main state space $Q$ described in \cref{sec:overall_protocol}. 
Note that this field is not changed by the protocol.
We denote $\QDetectCollision$ as the local state (sub-)space of $\DetectCollision$.
We use $\top \in \QDetectCollision$ to denote the state $\DetectCollision$ uses to indicate that a collision was found.
For all further details concerning this state space, we refer to \cref{sec:error_detection}.
And $\probationTimer \in [P_{max}]$ where $P_{max} = \cprobation \cdot \frac{n}{r} \cdot \log n$.

\paragraph{Algorithm}
The module, given in pseudocode as Protocol~\ref{pr:framework},
    works broadly as how we describe the soft reset mechanism in \cref{sec:overview:softreset}.
In each interaction, agents decrement their $\probationTimer$ (lines 1--2),
    which we use to decide whether to perform full or soft resets.
Then, if (and only if) two interacting agents have equal $\generation$, they update $\qDetectCollision$ using the $\DetectCollision$ module (lines 3--4; cf.\ \cref{sec:error_detection}):
    they then check if it generated the special error state $\top \in \QDetectCollision$, indicating that a collision (either in the ranking or an message inconsistency) was detected (though it may be either genuine, or only due to adversarial misinitialization),
        and initiate a soft or full reset depending on $\probationTimer$ (lines 5--8).
Whenever an agent in generation $i$ meets an agent in generation $i+1\;(\text{mod}\; 6)$ while its $\probationTimer$ is $0$, it adopts the other agent's generation
    and resets its other local state (lines 10--12).
However, if this happens while its $\probationTimer$ is positive,
    or if the other agent's generation differed by more than one,
    the agent triggers a full reset (line 13).

\begin{algorithm}[!ht]{$\StableVerify(u,v)$, where $u$, $v$  have $\role=\verifying$. \label{pr:framework}}
$u.\probationTimer \gets \max\cbc{0, u.\probationTimer - 1}$ /* decrement probation timers if possibile */
$v.\probationTimer \gets \max\cbc{0, v.\probationTimer - 1}$
if $u.\generation = v.\generation$ then /* same-generation verifiers execute $\ErrorDetection$ */
  $(u.\vState, v.\vState) \gets \ErrorDetection((u.\rank, u.\vState), (v.\rank, v.\vState))$
  for $i \in \cbc{u, v}$ with $i.\vState = \top$ do /* error detected */
    if $i.\probationTimer = 0$ then /* $i$ not on probation: soft reset */
      $i.\generation \gets i.\generation + 1 \pmod{6}$; $i.\vState \gets \qNaughtDetectCollision$; $i.\probationTimer \gets P_{max}$
    else $\TriggerReset(i)$ /* $i$ on probation: full reset */
  return

for $(i,j) \in \cbc{(u,v),(v,u)}$ with $i.\probationTimer = 0$ and $i.\generation \equiv j.\generation - 1 \pmod{6}$ do
  $i.\generation \gets j.\generation$; $i.\vState \gets \qNaughtDetectCollision$; $i.\probationTimer \gets P_{max}$ /* soft reset via epidemic */
  return

$\TriggerReset(u)$ /* generations are different, but no soft reset permissible: hard reset */
\end{algorithm}

\subsection{\texorpdfstring{$\ErrorDetection$}{ErrorDetection}}

\label{sec:error_detection}

\newcommand{\smap}{\ensuremath{\mathcal{S}}}

We now describe the submodule $\ErrorDetection$, 
    which implements the collision detection mechanism laid out in \cref{sec:overview:collision} while incorporating the space-time trade-off through the methods described in \cref{sec:overview:tradeoff}.
For the latter, recall that we use a partition of the rank space $[n]$ into $\ceil{\nicefrac{n}{r}}$ groups of size $\Theta(r)$.
Such a partition, in particular one where the groups have sizes in $\{r/2, \ldots, r\}$, always exists.
We assume that one such partition is encoded directly into the transition function via access to a function $\smap : [n] \to 2^{[n]}$ mapping ranks to the group containing the rank.
For any agent $u$, we define $r_u = |\smap(u.\rank)|$ as the size of $u.\rank$'s group and $u.\rank_r \coloneqq u.\rank \pmod{r_u}$ as the position of $u$'s rank within~$\mathcal{S}(u)$.

We outline the submodule here; more details, including the pseudocode for submodules used by it, are given in \cref{apx:error_detection}.

\paragraph{State Space \texorpdfstring{\&}{and} Notation}

\begin{wrapfigure}{R}{0.52\textwidth}
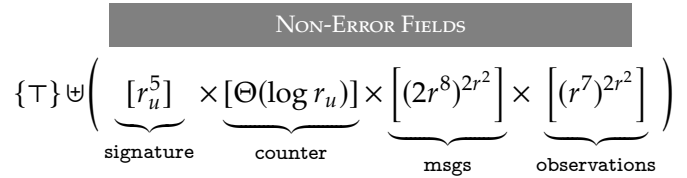

\centering
\[
\arraycolsep=0.5pt
\begin{array}{@{}ccccccccccccccccccc@{}}
&{\graybox{17.5em}{\mathstrut\textsc{\footnotesize Non-Error Fields}} }\\
\{\top\} &\uplus \bigg(  \underbrace{[r^5_u]}_{\texttt{signature}}  \times  
\underbrace{[\Theta(\log r_u)]}_{\texttt{counter}}
\times
\underbrace{\left[(2r^8)^{2r^2}\right]}_{\texttt{msgs}} \times 
\underbrace{\left[(r^7)^{2r^2}\right]}_{\observations} \bigg)
\end{array}
\]
    \caption{An overview of $\DetectCollision$'s state space. The overall state complexity is $2^{\Oh(r^2\log{r})}$.}
    \label{fig:state-space-detect}
\end{wrapfigure}
The state space of $\DetectCollision$  consists of the error state $\top$ and additional non-error states having following fields:
    $\tagVal$ is the message content for the messages an agent sends out,
    $\counter$ times the interactions until $\tagVal$ is refreshed,
    $\messages$ stores the messages as a sparse array of values in $[r_u^5]$ indexed by $(i,j)\in \smap(u.\rank) \times [2r_u^2]$,
    and $\observations$ stores the local message copies  the agent generated as a dense array of $2r_u^2$ values in $[r_u^5]$.  $\DetectCollision$ requites read-only access to the field $\rank \in [n]$ (\cref{sec:overall_protocol}).

Upon starting the protocol or being soft-resetted, the agent assumes the state $\qNaughtDetectCollision$, where $\tagVal$, $\counter$ and all values in $\observations$ are $1$, while $\messages$ gets initial value 1 for all indices $\smap(u)\times \{2(u.\rank_r -1)n+1,...,(2 u.\rank_r) n\}$ and $\bot$ otherwise.
For technical reasons, we restrict the state space such that if an agent $u$ has rank $i$ and stores the message indexed by $(i,j)$, the contents of $u.\messages [(i,j)]$ match $u.\observations[j]$.
This could only be violated if the process were initialized in this manner, so we can circumvent it by definition.

\paragraph{Algorithm}
$\DetectCollision$ works as follows.
Only agents belonging to the same group have non-trivial interactions (lines 1--2):
First, they check if they have either the same rank or a duplicate circulating message,
    and raise an error in that case (lines 3--4).
Next, they use three subroutines whose pseudocode is given in \cref{apx:error_detection}:
$\CheckMessageConsistency$ (Protocol~\ref{pr:consistency}) checks if any of contents of circulating messages in an agents' $\messages$ array is inconsistent with the contents stored in an $\observations$ array,
    and raises an error if this is the case.
$\UpdateMessages$ (Protocol~\ref{pr:updatemsg}) updates the contents of circulating messages having an agents' $\rank$ to its current $\tagVal$,
    updating the corresponding $\observations$ as well.
$\LoadBalance$ (Protocol~\ref{pr:loadbalance}) swaps messages between $u$ and $v$ such that, for each pair $i\in \smap(u)$ and $k \in [r^5]$, $u$ and $v$ both receive half of the messages governed by $i$ and with content $k$ held by $u$ and $v$.

\begin{algorithm}[!ht]{$\ErrorDetection(u,v)$\label{pr:error}}
if $\smap(u.\rank) \neq \smap( v.\rank)$
  return

if $u.\rank =v.\rank$ or ($\exists (i,j) \in \smap(u.\rank)\times [r_u^2]:$ $u.\messages [(i,j)]\neq \bot$ $\wedge\ v.\messages [(i,j)]\neq \bot$) then
  set local state of $u,v$ to $\top$ /* Rank shared, or two copies of same circulating message: obvious collision */
$\CheckMessageConsistency(u,v)$;  $\CheckMessageConsistency(v,u)$ /* may raise error */
$\UpdateMessages(u,v)$; $\UpdateMessages(v,u)$
$\LoadBalance(u,v)$
\end{algorithm}
\paragraph{Requirements} We give a formal description and proof of the properties that $\StableVerify$ requires of $\ErrorDetection$ in \cref{prop: error detection}. Informally, $\StableVerify$ requires that (i)~from any configuration with a rank collision, $\ErrorDetection$ produces a $\top$ symbol within $\Oh(\nicefrac{n}{r}\log{n})$ time, and (ii)~if all agents are initialized from $\qNaughtDetectCollision$ and the ranking is correct, $\ErrorDetection$ will never generate a $\top$.

\section{Overview of our Analysis}\label{sec:analysis}

In this final section we present the proof of our main result, \autoref{thm:framework},
    which serves double duty as an overview of the overall proof structure.
We then state its constituent lemmas (save for a strictly technical result) and defer their proofs of the lemmas to appendices.

\begin{proof}[Proof of \cref{thm:framework}]
Our statement on the state space follows immediately from \autoref{fig:state-space}.
The correctness proof has three components: 
    By \cref{lem:safe_set_is_safe} there exists a ``safe'' set of configurations $\safe$ which is stable and forms a strict subset of those configurations with a correct ranking.
    By \cref{lem:correctness_from_reset}, the protocol reaches a configuration in $\safe$ within $\Oh(\nicefrac{n^2}{r}\log n)$ interactions of a full reset being triggered, w.h.p.
    Finally, by \cref{lem:recovery},
        from an arbitrary configuration, the protocol either triggers a reset or reaches a configuration in $\safe$ within $\Oh(\nicefrac{n^2}{r})$ interactions, w.h.p.

Thus, from an arbitrary configuration w.h.p. our protocol reaches either a configuration in $\safe$ or a fully dormant configuration within $\Oh(\nicefrac{n^2}{r}\log{n})$ interactions.
If it reaches a fully dormant configuration, then it in turn reaches a configuration in $\safe$ within $\Oh(\nicefrac{n^2}{r}\log{n})$ w.h.p.
Once the protocol reaches a configuration from $\safe$ it remains correct forever.
\end{proof}

The definition of the set $\safe$ of safe configurations is somewhat delicate and involves internals of $\StableVerify$ as well as the reachability of configurations by $\DetectCollision$; we state it here, with the full proof of safety deferred to \cref{apx:lem:safe_set_is_safe}.

First we need some additional definitions.
For a configuration $C$ of $\ElectLeader$ where all agents are verifiers,
    let $\tilde{C}$ be the configuration of $\ErrorDetection$ obtained by keeping only the $\rank$ and $\ErrorDetection$'s state $\qVerifying$.
Further, let $\tilde{C}_0$ be the configuration obtained from $\tilde{C}$ by replacing $\ErrorDetection$'s state to $\qNaughtDetectCollision$, for each agent.
Finally, let $\tilde{C}_{i+}$ be the configuration obtained from $\tilde{C}$ by replacing $\ErrorDetection$'s state to $\qNaughtDetectCollision$ only for those agents which have $\generation = i$ in $C$.
\begin{lemma}[Safety]\label{lem:safe_set_is_safe}
    The set of configurations $\safe \subset \mathcal{C}$ where both of the following properties hold is \emph{safe} (i.e., it is impossible for the protocol to leave the set, and the protocol's output is correct):
    \begin{enumerate}[label=(\alph*),noitemsep,nolistsep]
        \item All agents are verifiers and the ranking is correct.
        \item There is an $i$ such that all $\generation$ fields are either $i$ or $i+1 \pmod{6}$,
            all agents with $\generation = i$ have $\probationTimer = 0$, and $\tilde{C}_{i+}$ is reachable from $\tilde{C}_0$ by $\ErrorDetection$.
    \end{enumerate}
\end{lemma}
\begin{proof}[Proof Sketch]
    This follows from careful case analysis of all possible transitions. 
    For ranks to change or an agent to change role requires a full reset to occur. 
    This in turn requires an agent with a positive probation timer to change generation, and so it must be started by an agent in generation $i+1$. 
    However, from the correctness of $\DetectCollision$ we obtain that no $\bot$ can be generated by an agent in generation $i+1$.
\end{proof}
The following lemma captures the fact that our protocol quickly reaches a safe configuration after a reset; we give its full proof in \cref{apx:correctness after reset}:

\begin{lemma}[Correctness after a full reset]\label{lem:correctness_from_reset}
    $\ElectLeader$ reaches a configuration in $\safe$ within $\Oh(\nicefrac{n^2}{r} \log n)$ interactions w.h.p.\ starting from a triggered configuration.
\end{lemma}
\begin{proof}[Proof Sketch]
    W.h.p. the following occurs.
    From a full reset, we reach a fully dormant configuration within $\Oh(n\log{n})$ interactions.
    Then, from such a configuration $\AssignRanks$ produces a correct ranking and becomes silent long before the $\countdown$ runs out after $\Theta(\nicefrac{n^2}{r}\log{n})$ interactions.
    When this occurs, agents move to become $\verifying$, keeping their correct ranks and will never raise a $\top$.
    The agents are all verifiers, the ranking is correct, all agents are in generation $0$ and so we are in $\safe$.
\end{proof}

Our recovery lemma is the most complex of the result listed here.
Because of this, we present its proof, itself using sub-lemmas, here,
    with the statement of said sub-lemmas and their proofs given in \cref{apx:lem:recovery}.
\begin{lemma}[Recovery]\label{lem:recovery}
    Starting from an arbitrary configuration, $\ElectLeader$ triggers a reset or reaches configuration in $\safe$ within $\Oh(n\log{n})$ interactions w.h.p.
\end{lemma}

\begin{proof}
    We define a hierarchy of nested subsets of the set $\cC$ of all configurations, such that any configuration will either quickly move to a smaller subset or a reset is triggered w.h.p.
    Formally, we define a hierarchy of configuration sets $\cC \eqqcolon \cC_0 \supset \cC_1 \supset \cC_2 \supset \cC_3 \supset \cC_4 \supset \cC_5$, where each set contains exactly those configurations where the following properties are true:
    \begin{itemize}[noitemsep]
        \item $\cC_1$: No agent is a resetter (i.e., all agents are rankers or verifiers).
        \item $\cC_2$: Additionally to $\cC_1$, no agent is a ranker (i.e., all agents are verifiers).
        \item $\cC_3$: Additionally to $\cC_2$, all $\generation$ values are equal.
        \item $\cC_4$: Additionally to $\cC_3$, all agents have $\probationTimer = 0$.
        \item $\cC_5$: Additionally to $\cC_4$, all $\rank$ fields are all distinct, i.e., the ranking is correct.
    \end{itemize}
    
    \Cref{lem:recovery:resetting} shows that from any configuration in $\cC_0 \setminus \cC_1$, the protocol will reach a configuration in $\cC_1$ within $\Oh(n \log n)$ interactions w.h.p.
    \Cref{lem:recovery:producing,lem:recovery:generations,lem:recovery:probation} then subsequently show that from any configuration in $\cC_i \setminus \cC_{i+1}$ (for $i \in \{1, 2, 3\}$),
        the protocol will trigger a reset or reach a configuration in $\cC_{i+1}$ within $\Oh(\nicefrac{n^2}{r})$ interactions w.h.p.\ 
        (in fact in $\Oh(n \log n)$ interactions for $i=2$).
    Lastly, \cref{lem:recovery:incorrect_output} shows that from any configuration in $\cC_4 \setminus \cC_5$, the protocol will trigger a reset within $\Oh(\nicefrac{n^2}{r} \log n)$ interactions w.h.p.
    As $\mathcal{C}_5\subset \safe$, this taken together means that from any configuration within $\Oh(\nicefrac{n^2}{r} \log n)$ either the population enters a safe configuration or a reset is triggered within $\Oh(\nicefrac{n^2}{r} \log n)$ interactions w.h.p.\ from \emph{any configuration}, as claimed.
\end{proof}

\newpage

\appendix 

\section{Toolbox}
\label{apx:toolbox}
\begin{lemma}[Technical Lemma for counting interactions in the population as a whole]
\label{lem: full population technical}
    For $t(n)\geq 4n\log{n}$ for all sufficiently large $n$, with probability $1-o(n^{-2})$ no agent will be involved in more than $\nicefrac{\alpha t(n)}{n}$ or fewer than $\nicefrac{t(n)}{\alpha n}$ within $t(n)$ interactions of the overall population for all $\alpha>7$.
\end{lemma}
\begin{proof}
    Each pair of interacting agents is chosen independently of all others,
        and each agent appears in an interaction with probability $\frac{2(n-1)}{n(n-1)}=\nicefrac{2}{n}$.
    So the number $X_i$ of interactions had by a specific agent $i$ has distribution $\Bin(t(n), \nicefrac{2}{n})$, with expectation $\nicefrac{2t(n)}{n}$.
    Applying the weak multiplicative Chernoff Bound,
    \begin{align}
    \Prob{X_i>\frac{\alpha t(n)}{n}}&\leq e^{-\frac{(\alpha-1)^2t(n)}{(2+\alpha-1)n}}
    \leq e^{\frac{36t(n)}{8n}}
    \leq n^{-4}
    \end{align}
    Then taking a union bound over such an event,
    \begin{align}
        \Prob{\bigvee_{i \in [n]}X_i>\frac{\alpha t(n)}{n}}&\leq \sum_{i \in [n]}\Prob{X_i>\frac{\alpha t(n)}{n}}
        \leq n^{-2}
    \end{align}
    For the other direction,
    \begin{align}
        \Prob{X<\nicefrac{t(n)}{\alpha n}}&\leq e^{\frac{-(1-\nicefrac{1}{2\alpha})^2t(n)}{n}}
        \leq e^{-\frac{9t(n)}{10n}}
        \leq n^{-\nicefrac{7}{2}}
    \end{align}
    Then applying a union bound as before we obtain the probability of this happening for any agent is $o(n^{-2})$.
\end{proof}
Finally, our protocol depends heavily on the use of \emph{one-way epidemics} to quickly spread information throughout the population.
In a two-way epidemic, one (or more) agents are \emph{marked}, and on every interaction between a marked and an unmarked agent, both become marked.
For our purposes, we simply require the following result, which we will use throughout. The statement is immediately from Lemma 2.9 of \cite{DBLP:journals/corr/abs-1907-06068}.
\begin{lemma}
\label{lemma:epidemic}
    There exists a constant $\cepi <7$ such that if $n$ epidemics are started at time $t$, there exists $t_1\leq \cepi n\log{n}$ such that all epidemics will infect all agents after $t_1$ interactions w.h.p.
\end{lemma}

\section{Derandomizing the Transition Function}
\label{apx:derandomization}
Recall that we made the simplifying assumption that the agent are able to sample random values.
This was only done for an cleaner description of the protocols as typically one assumes a deterministic transition function.
In the following, we will briefly sketch how derandomization technique that can be used to adapt our protocols such that they only require the randomization of the scheduler.
In our protocols, on several occasions, an agent must sample a value from $[r^5]$ uniformly at random.
Here, $r > 0$ is a fixed parameter encoded in the transition function. 
However, the agents do not need to sample such a value in every iteration.
In the worst case, an agent needs to sample such a value every $\alpha\log{r}$ of its own interactions. 
The value $\alpha$ is a large tunable constant determined in the analysis.
\begin{lemma}[Derandomization of $\delta$]
\label{lemma:derandomization}
Let $\mathcal{P}$ be a population protocol with a (probalistic) transition function $\delta$ that fulfills the following two properties: 
\begin{enumerate}
    \item  each agent $v \in V$ must sample a random value $x \in [N]$ (almost) u.a.r.
    \item Between two interactions where $v$ must sample a value, it is activated at least $\alpha \log r$ times.
\end{enumerate}
Then, $\mathcal{P}$ can be implemented with a determistic transition function $\delta'$ using only randomness from the scheduler. 
The blowup in the number of states is within $O(N\log N)$.
\end{lemma}
We can implement this sampling in the following way:
Each agent $v \in V$ stores three additional fields.
First, a variable $v.\Coin \in \{0,1\}$, a binary array $v.\Coins \in \{0,1\}^{\log N}$ of length $\log N$ and a cyclic counter $v.\CoinCount \in \mathbb{Z}_{\log N}$. 
Storing these values requires $\log N + \log\log N + 1$ additional bits and blows up the state space by a factor of $2N\log N$.
On every interaction, an agent $v$ flips the value of $v.\Coin \in \{0,1\}$ to its complement, i.e., it sets
\begin{align}
    v.\Coin \gets 1 - v.\Coin
\end{align}
Also, on every interaction $v.\CoinCount$ is increased by one and wraps around to $0$ once it surpasses it maximal value $\log n$.
Formally, we have
\begin{align}
    v\CoinCount \gets v\CoinCount +1 \pmod {\log N} 
\end{align}
Further, the array $\Coins_v$ stores the last $\log N$ coins observed in the last $\log N$ interactions.
That is, on an interaction with agent $w \in V$, agent $v$ observes $w.\Coin$ and sets:
\begin{align}
    i &\gets v.\CoinCount\\
    v.\Coins[i] &\gets  w.\Coin
\end{align}
These $\log N$ values in $v.\Coins$ encode the desired random values.

Finally, we show why this process produces the desired random numbers.
Clearly, after $\log N$ activations, the values in $v.\Coins$ have completely replaced by new values.
Thus, two consecutive samples are independent of one another.
It remains to show that the samples have the correct probability.
\textcite{berenbrink2019tight} have shown that after $O(n\log N)$ global interactions, there are $(\nicefrac{1}{2} \pm\nicefrac{1}{10\log N}) \cdot n$ agents with either value of the coin w.h.p.
Thus, the probability that the last $\log N$ coins stored in $v.\Coins$ form the bit representation of a value $x \in [N]$ is between $(\nicefrac{1}{2} - \nicefrac{1}{10\log N})^{\log N} \geq \frac{1}{2N}$ and $(\nicefrac{1}{2} + \nicefrac{1}{10\log N})^{\log N} \leq \frac{2}{N}$ as desired. 

Thus, we only need to make sure that there are $O(N \log N)$ global inter interactions before each sampling, the lemma follows.
Here, the second property of the protocol comes into play.
We require $\alpha \log r$ interaction between each sample.
Choosing $\alpha \gg 5$ large enough and applying \cref{lem: full population technical} shows that there enough time to the coins to converge to the right distribution before each sampling, w.h.p.
\section{The PropagateReset Protocol}\label{apx:PropagateReset}
In this section we describe the reset mechanism of \texorpdfstring{\textcite{DBLP:conf/podc/BurmanCCDNSX21}}{Burmann et al.} which we use in our protocol. We state its pseudocode in Protocol~\ref{pr:propagatereset}.
We will need to know just a bit of its inner workings:
Resetting agents have their $\role$ field set to $\resetting$;
    if $\role \neq \resetting$, we call an agent \emph{computing}.
Any configuration containing a computing agent is \emph{partially computing}.
Agents in the $\resetting$ role have
    a field $\resetCount \in \cbc{0, \ldots, R_{max}}$ (with $R_{max} = \Theta(\log n)$),
    and a field $\delayTimer \in \cbc{0, \ldots, D_{max}}$ (with $D_{max} = \Theta(\log n)$, and $\delayTimer = D_{max}$ whenever $\resetCount > 0$).
To trigger a reset, an agent is made a resetting agent according to $\TriggerReset$ (Protocol~\ref{pr:triggerreset});
    an agent in that state is \emph{triggered}, and any configuration containing such an agent is \emph{triggered}.
While $\resetCount$ is non-zero, an agent may ``infect'' computing agents to also become resetting.
This counter is decremented in every interaction;
    once it hits $0$, an agent becomes \emph{dormant},
    upon which they initialize $\delayTimer$, decrement it in every interaction, and become computing again when the timer hits $0$ (with users of $\PropagateReset$ describing how computing agents are initialized by giving a routine \textsc{Reset}).
Computing agents additionally ``wake up'' dormant agents via epidemic.
Any partially computing configuration
    which is the first such configuration reachable from a fully dormant configuration
    is called an \emph{awakening configuration}.

The following results capture almost everything we need to know about $\PropagateReset$:
\begin{lemma}[Combination of Lemmas 3.2 and 3.3 in \cite{DBLP:journals/corr/abs-1907-06068}]\label{lem:triggered-to-dormant}
    Let $R_{\max} = 60 \log n$ and $D_{max} = \Omega(\log n)$.
    Starting from a triggered configuration, we reach a fully dormant configuration after $\Oh(n \log n)$ interactions with probability $1 - \Oh(1/n)$.
\end{lemma}

\begin{theorem}[Theorem 3.4 in \cite{DBLP:journals/corr/abs-1907-06068}]\label{thm:propagate-reset}
    Let $R_{max} = 60 \log n$ and $D_{max} = \Omega(\log n+ R_{max})$. Starting from a partially-triggered configuration, we reach an awakening configuration in at most $O(D_{max} n)$ interactions with probability at least $1 - \Oh(1/n)$.
\end{theorem}
We can simplify these statements as follows
\begin{corollary}[Correctness of \PropagateReset, simplified]
\label{lemma:propagatereset}
There is a population protocol $\PropagateReset$ with $\Theta(\log n)$ states and the following properties:
\begin{enumerate}[noitemsep,topsep=0pt]
    \item Starting configuration where an agent calls \TriggerReset, we reach  fully dormant configuration in at most time $O(\log n)$ w.h.p.
    \item Starting from a fully dormant configuration, we reach an awakening configuration in at most time $O(\log n)$ w.h.p.
    \item  Starting from a configuration with some {propagating} agent, we reach a computing configuration in at most $O(\log{n})$ time w.h.p. unless another reset is called.
\end{enumerate}
\end{corollary}

\begin{algorithm}[!ht]{$\PropagateReset(u,v)$ \label{pr:propagatereset}, where $u.\role$ must be $\resetting$}
if $u.\resetCount > 0$ and $v.\role \neq \resetting$ then
  $v.\role \gets \resetting$; $v.\resetCount \gets 0$; $v.\delayTimer \gets D_{max}$
  
if $v.\role = \resetting$ then
  $a.\resetCount \gets b.\resetCount \gets \max\cbc{a.\resetCount - 1, b.\resetCount - 1, 0}$ 
  
for $(i, j) \in \cbc{(u, v), (v, u)}$ with $i.\role = \resetting$ and $i.\resetCount = 0$ do
  if $i.\resetCount$ just became 0 then
    $i.\delayTimer \gets D_{max}$
  else
    $i.\delayTimer \gets i.\delayTimer - 1$
    
  if $i.\delayTimer = 0$ or $j.\role \neq \resetting$ then
    execute $\textsc{Reset}(i)$
\end{algorithm}

\begin{algorithm}[!ht]{$\TriggerReset(u)$ \label{pr:triggerreset}}
$u.\role \gets \resetting$
$u.\resetCount \gets R_{max}$
$u.\delayTimer \gets D_{max}$
\end{algorithm}

The routine $\textsc{Reset}$ which we use to (re-)initialize agents is given by Protocol \ref{pr:reset}.
\begin{algorithm}[!ht]{$\textsc{Reset}(u)$.\label{pr:reset} This routine, used exclusively in $\PropagateReset$, initializes $u$'s state.}
$u.\role \gets \producing$; $u.\pState \gets \qNaughtRanking$; $u.\countdown \gets C_{max}$
\end{algorithm}

\section{Ranking}
\label{apx:ranking}

In this section we present the parametrized (\emph{non}-self-stabilizing) ranking protocol $\Ranking$, which satisfies the following Lemma.
\begin{lemma}[Fast and Silent Ranking]
    \label{lem:fast-and-silent-ranking}
     $\Ranking$ assigns unique rank in $[n]$ to each agent in time $\crank\cdot\nicefrac{n}{r}\cdot\log n$ w.h.p, starting from any dormant configuration.
    $\Ranking$ uses $2^{\Oh(r \log n)}$ states.
\end{lemma}
In the time-optimal regime, one could use other previously established approaches to compute a ranking (in a way that is not self-stabilizing):
    In the protocol of \cite{DBLP:conf/podc/BurmanCCDNSX21},
        agents choose one of $\Oh(n^3)$ names at random.
    They then broadcast these names, storing the entire set of seen names, and obtain ranks from this set (as the used names are unique w.h.p.); this requires $\Oh(n\log n)$ bits and $\Oh(n\log n)$ interactions w.h.p.
    Alternatively, one could adapt a protocol in \cite{DBLP:conf/opodis/GasieniecJLL21} which assigns unique labels in $[2n]$ quickly, so that storing the set of used labels takes only $\Oh(n)$ bits.
However, for the regime using $\tilde{\omega}(n^{3/2})$ interactions,
    given our collision detection protocol, this would make rank assignment the bottleneck for the size of the state space.
In particular, we would not obtain a protocol using $o(n^2)$ interactions with sub-exponential state space size.

The fundamental technique we use is a standard one: we assign to each agent a unique ``label'' drawn from a space much larger than $[n]$, our agents communicate to learn which labels exist and then pick ranks based on that information.
However, we need to avoid agents having to exchange names directly with the rest of the population and instead consolidate the information to be sent more efficiently.
We achieve this by delegating the assignment of labels to $r$ unique ``deputies'' and thus agents are only required to know how many labels each deputy has given out in order to learn the full set of labels.
\medskip

More explicitly, in order to nominate a \emph{sheriff}, we use a black-box fast leader election protocol that is \emph{not} self-stabilizing and functions from an awakening configuration (see \autoref{apx:simple_leader_elect}). 

The sheriff, once nominated, has a collection of $r$ badges for some parameter $1\leq r<\nicefrac{n}{2}$. Each badge has a unique badge number from ~$[r]$. 
As soon as a sheriff  meets a \emph{recipient} (i.e. an agent who was not elected sheriff), the sheriff promotes the recipient to a sheriff and gives them half of its badges. 
When a sheriff has only a single badge left, it becomes a \emph{deputy}. 

A deputy has a unique \emph{id} from $[r]$, for which we use its badge id, and maintains a counter initialized to $1$. 
Upon meeting a recipient, the deputy increases its counter and assigns the recipient a temporary label $(i,j)$, where $i$ is
the deputy's id, and $j$ the new value of the deputy's counter.
After assigning the temporary label, the deputy broadcasts the total number of labels it has given out. 
To implement this broadcast, every agent continuously forwards the largest observed values of all the deputies' counters.

Once an agent hears that $n$ labels have been given out, i.e., the sum of all counters is $n$, it knows the set $\{(i,j) \mid i \in [r], j \in [c_i]\}$ of assigned labels, and gives itself the rank in lexicographic ordering.

After assigning itself a rank, an agent discards its remaining states.
However, this cannot happen immediately, as it could lose some of the states required for broadcasting the labelling information.
In other words, this could kill the broadcast partway and leave some agents stuck in the ranking phase, as they would never learn that $n$ labels are assigned. 
Therefore, we put the agents to sleep until the broadcast has likely finished. 
Upon waking up, they pick a rank and delete the rest of their state space.
\medskip

\noindent We will now describe the state space in more detail. An agent is in one of six types being either: in leader election, a sheriff, a deputy, a recipient, a sleeper or ranked. These have some unique fields, as well as some shared fields. Uniquely,
\begin{itemize}[noitemsep]
    \item an agent in leader election has all the fields of the leader election algorithm,
    \item a sheriff has $\lrank, \hrank  \in [r]$, with $\lrank < \hrank$. These determine the low and high end (inclusive) of the range of badges the sheriff may distribute and are initialized as $\lrank =1$ and $\hrank =r$,
    \item a deputy has an $\id  \in [r]$ and a $\counter  \in [cn/r]$ determining how many labels (including their own) they have given out,
    \item a recipient has a $\lbl  \in [r]\times[cn/r]\cup\{\bot\}$ which records the label they have received from a deputy and from which deputy. This is initialized to $\bot$ representing none,
    \item a sleeper has a timer $\sleep \in [\csleep  \log{n}]$ which is initialized to $1$ and a $\lbl  \in [r]\times[cn/r]\cup\{\bot\}$ which is inherited from the value as a recipient.
\end{itemize}
All non-leader election, non-ranked types shared the following field.
\begin{itemize}
    \item $\channel  \in \{0,...,cn/r\}^r$. These are $r$ channels that store the highest value each deputy has given out.
\end{itemize}
All  types shared the following field:
\begin{itemize}
    \item $\rank \in [n]$. This is the rank the agent currently believes itself to be. This is initialised to $1$ and update only when agent becomes ranked
\end{itemize}
The algorithm then proceeds as described in Protocol~\ref{pr:ranking}, and makes use of $\Oh(L(n))+n+(r^2+cn+cn+1+(cn+1)(\csleep\log{n}))(\nicefrac{cn+r}{r})^r=2^{\Oh(r\log{n})}$ states.

\begin{algorithm}[!ht]{$\Ranking(u,v)$\label{pr:ranking}}
if $u$ or $v$ is in leader election then $\SheriffElection(u,v)$
else
  if $u$ or $v$ in sleep then $\Sleep(u,v)$
  else if $u$ or $v$ is a sheriff and the other a recipient then
    $\Deputize(u,v)$
  else if $u$ or $v$ is a deputy and the other is a recipient without a label then
    $\Labelling(u,v)$
  for each $i \in [r]$ do
    $u.\channel [i], v.\channel [i]\leftarrow \max(u.\channel [i], v.\channel [i])$
  if $\sum_{i=1}^r u.\channel [i]=n$ then 
    $u$ and $v$ become sleepers
\end{algorithm}

\begin{algorithm}[!ht]{$\SheriffElection(u,v)$\label{pr:sheriff}}
if both $u$ and $v$ are in leader election then
  update states according to black-box leader election
else if exactly one of $u$ and $v$ is in leader election then
  the other becomes a recipient
\end{algorithm}

\begin{algorithm}[!ht]{$\Deputize(u,v)$\label{pr:deputize}}
Let $w$ be the sheriff and $x$ be the recipient
$x$ becomes a sheriff
$x.\hrank \gets w.\hrank $
$w.\hrank \gets \lfloor (w.\hrank +w.\lrank ) / 2 \rfloor$
$x.\lrank \gets w.\hrank + 1$
for $z$ in $\{x, w\}$
  if $z.\hrank =z.\lrank $ then
    $z$ becomes a deputy
    $z.\id \gets z.\lrank $
    $z.\counter \gets 1$
    $z.\channel [z.\id ]\gets 1$
\end{algorithm}

\begin{algorithm}[!ht]{$\Labelling(u,v)$\label{pr:label}}
if $\sum_{i=1}^r u.\channel \geq r$ then
  Let $w$ be the deputy and $x$ the recipient
  if $w.\counter <cn/r$ then
    $w.\counter \leftarrow w.\counter +1$
    $w.\channel [w.\id ]\leftarrow w.\counter $
    $x.\lbl =(w.\id , w.\counter )$
\end{algorithm}

\begin{algorithm}[!ht]{$\Sleep(u,v)$\label{pr:sleep}}
Let $x$ be a sleeping agent and $w$ the other
if $w$ is ranked then
  $x$ becomes ranked and updates their $\rank$ based on its $\lbl $ and $\channel $ fields
else if $x.\sleep =\csleep  \log{n}$ then
  $x$ and $w$ become ranked and updates their $\rank$ based on their $\lbl $ and $\channel $ fields
else
  $w$ becomes a sleeper
\end{algorithm}

\subsection{Proof of \texorpdfstring{\autoref{lem:fast-and-silent-ranking}}{Main Lemma for AssignRanks}}
Our strategy is to show that beginning from a dormant configuration the following ``correct'' execution occurs w.h.p.. 
\begin{itemize}[noitemsep]
    \item All agents begin dormant
    \item An agent wakes up and propagates this to all other agents
    \item A unique sheriff is elected by the black-box leader election protocol
    \item That sheriff selects $r$ deputies.
    \item The deputies distribute labels until $n$ labels have been given out in total, continuously broadcasting which labels have been used.
    \item All agents go to sleep
    \item All agents wake up and choose a rank based on the position of their label amongst those distributed.
\end{itemize}
In each case we will show that for each in turn, the step occurs sufficiently quickly and correctly w.h.p. 
Throughout the proof we will take the $\Ranking$ subprotocol in isolation, other than when initially waking up from $\PropagateReset$.
We will make use of the following definition.
\begin{definition}[Ruled Population]
    The population is \emph{ruled} if there exists a single sheriff with $\hrank =r$, $\lrank =1$ and all entries of its channel field set to $0$, and all other agents are in a terminal state of the leader election protocol.
\end{definition}
\begin{lemma}
    There exists $a>0$ such that beginning from a dormant configuration, within $an\log{n}$ interactions either a unique sheriff is elected and all other vertices remain of the leader election type (i.e. the population is ruled) w.h.p.
    \label{lemma:reset to ruled}
\end{lemma}
\begin{proof}
    As expressed in the proof of \autoref{thm:propagate-reset} from \cite{DBLP:journals/corr/abs-1907-06068}, from a fully dormant configuration an agent executes $\Reset$ within $\nicefrac{n}{2}D_{\text{max}}$ interactions by the pigeon hole principle and becomes a ranker.
    This role propagates as an epidemic and so w.h.p.\ all agents becoming rankers within $\cepi n \log{n}$ interactions.
    From here, this follows immediately from the correctness and run-time of a state of the art leader-election protocol, and the fact that any sheriff elected by the leader election protocol is initiated to have a full roster of badges from $\{1,...,r\}$ and its channel field all set to $0$.
\end{proof}

We can make the following observations.

\begin{observation}\label{obs:prot}\mbox{}\\[-2ex]
\begin{enumerate}[label=(\alph*),nolistsep,noitemsep]
    \item Any agent that leaves the black box leader election protocol to become a deputy, recipient or sheriff begins with a channel field that is either all $0$ or equal to that of the agent who spurred the change. \label{obs:Agents initialize to 0}
    \item If for all agents $\channel [i]=0$, then the only interaction that can produce an agent with $\channel [i]=1$ is the creation of a deputy with $\id $ $i$, in which case it increases by $1$.\label{obs:Only a deputy can change a channel from 0 to 1}
    \item The maximum value of $\channel [i]$ across all agents can only increase if a deputy with $\id $ $i$ gives out a label (possibly to itself). \label{obs: Only a deputy can increase the maximum value of a channel}
    \item If there is only a single sheriff or deputy with badge $z$ then no other agent can be created with badge $z$.\label{obs:Ids are unique}
    \item If there is a sheriff with badges $\{x,...,x+y\}$ then, for any $z \in \{x,...,x+y\}$, at any future time either: there is a sheriff with badge $z$, there is a deputy with badge $z$ or a deputy with badge $z$ has gone to sleep. \label{obs:Ids are conserved}
\end{enumerate}
\end{observation}

Now we shall show that all the deputies will be elected sufficiently quickly.
More precisely,
\begin{definition}[Quorate Population]
    The population is \emph{quorate} if there exist exactly $r$ deputies each with a unique $\id $ from $[r]$ and all other agents are either leader-election or unlabelled recipients. Furthermore, we require that the maximum value of each entry of the $\channel $ field is $1$ and that this value is held by the deputy with the matching $\id $ and is equal to their $\counter$ field.
\end{definition}
Combining this definition with the observations we obtain,
\begin{lemma}
     Suppose the population is ruled.
    Then, there exists a constant $b>0$ such that  the population becomes quorate within $bn\log{n}$ interactions w.h.p. 
    \label{lemma:ruled to quorate}
\end{lemma}
\begin{proof}
    By Observations 
    \ref{obs:prot}-\ref{obs:Agents initialize to 0},
    \ref{obs:prot}-\ref{obs:Only a deputy can change a channel from 0 to 1} and
    \ref{obs:prot}-\ref{obs: Only a deputy can increase the maximum value of a channel}, 
    as well as the definition of a state being ruled, we must have that no agent can have a channel field summing to $r$ or higher unless $r$ deputies have been created or a deputy has given out a label to another agent. However, $\Labelling(u,v)$ prevents the distribution of non-deputy labels until at least one agent has a channel field summing to $r$. Then by observation \ref{obs:prot}-\ref{obs:Ids are unique} combined with the uniqueness of the initial sheriff we must have that all deputies are unique. Thus, at the time the first label is distributed to a non-deputy agent all $r$ deputies exist and have unique $\id $s. Furthermore, all deputies must have a counter equal to $1$, the channel associated with their $\id $ equal to $1$ and no agent can have a channel field with that entry greater than $1$ by Observations 
    \ref{obs:prot}-\ref{obs:Agents initialize to 0},
    \ref{obs:prot}-\ref{obs:Only a deputy can change a channel from 0 to 1}, and
    \ref{obs:prot}-\ref{obs: Only a deputy can increase the maximum value of a channel}.

\newcommand{\pr}[1]{\Prob{#1}}
\newcommand{\E}[1]{\Exp{#1}}
\newcommand{\Deputy}[1]{\ensuremath{\mathsf{Deputy}_{#1}}}

    For ease of discussion, without loss of generality, we will assume that $\log_2{r} \in \mathbb{N}$, i.e., we assume $r$ is a power of $2$. 
    It thus suffices to show that all $r$ deputies will be elected within $bn\log{n}$ interactions w.h.p. 
    By the previous part and Observation \ref{obs:prot}-\ref{obs:Ids are conserved} we must have that at any point after our initial state with one sheriff elected and all others still in the leader-election phase, for each $i \in \{1,...,r\}$ either a deputy with badge $i$ has existed or there is currently a  sheriff with badge $i$ in its pool. 
    If a deputy with badge $i$ does not currently or formerly exist, the sheriff holding badge $i$ must have interacted with non-sheriff/deputy agents at most $\log_2{r}-1$ times.

    In the following, we let $\Deputy{i,t}$ be the event that a deputy with $\id $ $i$ existed within $t$ interactions of the election of a unique sheriff. We will show that 
    \begin{align}
        \pr{\Deputy{i,t}} \geq 1 - o(n^{-3}) \label{eq:deputy}
    \end{align}
    Thus, taking a union bound, all $r$ deputies exist before $bn\log{n}$ interactions have elapsed w.h.p.

    Therefore, it remains to prove \autoref{eq:deputy}. For a fixed badge $i$, let $v_t(i) \in V$ be the agent that holds badge $i$ in step $t$. Further, let  $Y_t(i)$ be the total number of badges $v_t(i)$'s pool. Whenever $v_t(i)$ interacts with an agent that is neither a sheriff nor a deputy, we have $Y_{t+1}(i) = Y_{t+1}(i)/2$. Recall that $Y_{t+1}(i)$ is always even as $r$ is a power of two.
    Otherwise, if $v_i(t)$ does not interact with a non-sheriff/non-deputy, we have $Y_{t+1}(i) = Y_{t+1}(i)$. 
    Let now $X_t(i)$ count the interactions between the agent holding badge $i$ and non-sheriffs/non-deputies within $t$ interactions, then $$Y_t(i) = r \cdot 2^{-X_t(i)}.$$
    Therefore,
    $$ \pr{\Deputy{i,t}} \geq 1 -\pr{X_{t}<\log_{2}{r}-1}.$$
    Since there are at most $\frac{n}{2}$ deputies/sheriffs as $r<n/2$, we have that the probability that on any given interaction the sheriff would interact with a non-sheriff/deputy agent is at least $\frac{1}{n}$. 
    Thus, the probability a deputy with $\id $ $i$ has not existed within $t$ interactions is dominated by the probability $\Prob{Z<\log_2{r}}$ for $Z\sim Bin(t,\nicefrac{1}{n})$. 
    For $t=bn\log{n}$ (and so $\Exp{Z}=b\log{n}$) and $\delta=1-\nicefrac{\log_2{r}}{b\log{n}}$ with $b$ sufficiently large we have, by a multiplicative Chernoff bound
    \begin{align*}
    \Prob{\Deputy{i,t}}
   &\ge 1-\Prob{X_t(i)<\log_{2}{r}-1}
    \geq 1-\Prob{Z_{t}<\log_{2}{r}}
    \geq 1-\Prob{Z_t<(1-\delta)\Exp{Z}}
\\ &=1-\exp\left(\frac{-(1-\nicefrac{\log_2{r}}{b\log{n}})^2b\log{n}}{2}\right)
    \geq 1-\exp\left(\frac{-(1-\nicefrac{\log_2{e}}{b})^2b\log{n}}{2}\right)
\\ &=1-n^{\frac{-(1-\frac{\log_2{e}}{b})^2b}{2}}
    \end{align*}
    For $b$ sufficiently large, we get that $\Prob{\Deputy{i,t}} \geq 1-o(n^{-3})$ proving \autoref{eq:deputy}. 
\end{proof}
Once the population is quorate, we can then show that every agent receives a label within a relatively short period of time. We make use of the following technical definition:
\begin{definition}[Valid/complete $\channel$ Field]
    For a configuration with exactly one deputy for each $\id  \in \{1,...,r\}$ and $v_i$ the deputy with $\id $ $i$, an agent $u$ has a valid $\channel $ field if for all $i\in \{1,...,r\}$ $u.\channel [i]\leq v_i.\counter $ or if it is in the leader-election phase. Additionally, if $u$ is a deputy we require that $u.\counter =u.\channel [u.\id ]$.
    
    Further, we say that an agent has a complete $\channel $ field if $\forall i \in \{1,...,r\}$ $u.\channel [i]=v_i.\counter $.
\end{definition}

Observe that unless at least one agent has become ranked, if all agents have valid $\channel $ in a given interaction this is conserved in the next interaction.
\begin{definition}[Well-labelled Population]
    The population is \emph{well-labelled} if all agents have a unique temporary label either concretely as a recipient or sleeper, or implicitly as a deputy. Furthermore, the maximum value of $\channel [i]$ is equal to the highest temporary label given out by deputy $i$.
\end{definition}
\begin{lemma}
    Beginning from a quorate state there exists $d>0$ such that the population is well-labelled within $\frac{dn^2\log{n}}{r}$ interactions w.h.p.
    \label{lemma:quorate to well-labelled}
\end{lemma}
\begin{proof}
    Since all agents have valid $\channel $ fields and there are $r$ deputies with their $\counter $ fields at $1$, by the condition in $\Labelling(u,v)$ one agent must gain a complete  $\channel $ field before any labelling can occur. 
    The maximum value in each subfield of the $\channel $ is propagated by an epidemic, and so by standard results, \textit{at least one} deputy will have a complete $\channel $ field within $d'n\log{n}$ interactions w.h.p.
    
    At this point, the agent with a complete $\channel $ field has a $\channel $ field that sums to $r$. 
    
    Since the maximum value in each subfield of the $\channel $ field is propagated by epidemic, again by standard results \textit{all} deputies will have a $\channel $ field with sum greater than or equal to $r$ within $\hat{d}n\log{n}$ interactions for some constant $\hat{d}>0$ w.h.p.
    
    After this point, if an unlabelled recipient or leader-election agent interacts with a deputy either the deputy has given out all of its labels or the agent receives a label. 
    Since there are $r$ deputies each with $cn/r$ labels to distribute, if there is an unlabelled agent then there are at least $(1-\nicefrac{1}{c})r$ deputies with free labels.
    
    Thus, in each interaction any given previously unlabelled agent receives a label with probability at least $\frac{2(1-\nicefrac{1}{c})r}{n(n-1)}$. 

    \medskip
    \noindent Let $\Un$ be the event that after $t$ interactions from all deputies having a $\channel $ field that sums to at least $r$, and there are still unlabelled agents. By a union bound,
        \[\Prob{\Un} \leq n\bc{1-\frac{2(1-\nicefrac{1}{c})r}{n(n-1)}}^{t}.\]
    Then as $1-x\leq e^{-x}$ and taking
        $t = 4 \log{n} \cdot \frac{n(n-1)}{2(1-\nicefrac{1}{c})r}\leq \frac{2n^2\log{n}}{(1-\nicefrac{1}{c})r}$,
    \[\Prob{\Un}
        \leq n \exp\bc{-t\cdot \frac{2(1-\nicefrac{1}{c})r}{n(n-1)}}
        \leq n\cdot e^{-4\log n}
        = n\cdot n^{-4} = n^{-3}
        =o(n^{-2}).\]
    Thus, for $d>d'+\hat{d}+\frac{2}{(1-\nicefrac{1}{c})}$ the claim holds.
\end{proof}
\begin{lemma}
\label{lemma:well-labelled to ranked}
    From the interaction where the population becomes well-labelled, there exists $f,\csleep  >0$ such that all agents will go to sleep and wake up within $fn\log{n}$ interactions and the last agent will go to sleep before the first one wakes up w.h.p. 
    Furthermore, if this occurs, every agent that wakes up will adopt a unique rank.
\end{lemma}

\newcommand{\FirstWake}{\ensuremath{\mathsf{FirstWake}}\xspace}

\begin{proof}
    Since new labels cannot be distributed, and the maximum value in each subfield of $\channel $ propagates according to an epidemic by standard results, there exists $f'>3$ such that all agents will have a complete $\channel $ field within $f'n\log{n}$ interactions w.h.p. 
    An agent goes to sleep the interaction it gets a complete $\channel $ field and will not wake up until some agent has had at least $\csleep  \log{n}$ interactions. 
    Let $\FirstWake$ be the random variable describing the number of interactions after the first agent goes to sleep the first agent wakes up. 
    $\FirstWake$ is clearly greater than the number of interactions before any agent takes $\csleep  \log{n}$ interactions after the first agent goes to sleep.
    \newcommand{\Steps}{\ensuremath{\mathsf{Steps}_{i}}\xspace}
    Let \Steps be the number of interactions performed by agent $i$ within $f'n\log{n}$ of the first agent getting a complete $\channel $ field. We can see that $\Steps\sim \Bin(f'n\log{n},\nicefrac{2}{n})$ and thus we can apply a Chernoff bound to obtain an an upper bound on the probability that too many interactions occur. Taking $\csleep  =3f'$ and $\delta=2$
    \[\Prob{\Steps<\csleep  }=\Prob{\Steps>(1+\delta)\Exp{\Steps}}
         \leq \exp(\frac{-\delta^2\Exp{\Steps}}{2+\delta})
         =\exp(-f'\log{n})=n^{-f'}.\]
    Thus, taking a union bound and $f'>4$,
    \begin{align}
        \Prob{\FirstWake}\leq n^{1-f'}\leq n^{-3}=o(n^{-2})
    \end{align}
    Thus, w.h.p.\ we get that all agents go to sleep before the first agent wakes up. Now since waking up only requires interacting $\csleep  \log{n}$ times we must have that one agent wakes up within $\csleep n\log{n}$ interactions by the pigeonhole principle with certainty. Thus, taking $f=\csleep$ gives the first part of the claim. Furthermore, since every agent has a complete and valid $\channel $ field they are aware of the label of every other agent. Thus all agents have the same information and so any pre-agreed bijection from the set of distributed labels to the ranks from $[n]$ will produce a unique ranking.
\end{proof}
Now tying it all together, we can prove that $\Ranking$ satisfies \autoref{lem:fast-and-silent-ranking}.

\begin{proof}[Proof of \autoref{lem:fast-and-silent-ranking}]
    Beginning, from a fully dormant configuration, for the protocol not to then produce a correct ranking within this time one of the following must fail:
    \begin{itemize}[noitemsep]
        \item The population does not become ruled within $an\log{n}$ interactions, which occurs with low probability (\cref{lemma:reset to ruled}).
        \item The population does not become quorate within $bn\log{n}$ interactions of being ruled, which occurs with low probability (\cref{lemma:ruled to quorate}).
        \item The population does not become well-labeled within $d\frac{n^2\log{n}}{r}$ interactions of being quorate, which occurs with low probability (\cref{lemma:quorate to well-labelled}).
        \item All agents do not uniquely pick a ranking within $fn\log{n}$ interactions of the last label being given, which occurs with low probability (\cref{lemma:well-labelled to ranked}).
    \end{itemize}
    Since each of these events fail to occur w.h.p., by taking a union bound we have that a unique ranking is obtained w.h.p. within $(a+b+d+f)\frac{n^2\log{n}}{r}$.
    Furthermore, once an agent is ranked the $\Ranking$ subprotocol can no longer update the $\qRanking$ component of its state, and thus has become silent.
    Finally, the state complexity follows immediately from the definition of the sub-protocol. Thus, the claim holds.
\newcommand{\MaxInteractions}[1]{\ensuremath{\mathsf{MaxInteractions}_{#1}}}
\end{proof}
\subsection{Non Self-stabilizing Leader Election from Awakening Configurations}
\label{apx:simple_leader_elect}

In this subsection, we describe $\FastLeaderElect$, a simple protocol that elects a leader w.h.p. from an awakening configuration.
We require this protocol because most non-self-stabilizing leader election protocols from the literature are designed to be started from configuration where all agents have the same state.
This is not the case for us, because even if we triggered a reset and are in an awakening configuration, not all nodes are in the same state.
In particular, some agents may \emph{wake up} much later than others and the protocol must be able to deal with this fact.
While many existing protocols can be adapted to work with awakening configurations, we present a simple protocol to be self-contained. 
Our protocol $\FastLeaderElect$ uses $O(n^3) \in 2^{O(\log n)}$ states and w.h.p. elects a unique leader in time $O(\log n)$ w.h.p.
The number of states is large if compared with other non self-stabilizing protocols but for our specific purpose, it is reasonable.
In particular, for every possible choice of $r \in [\nicefrac{n}{2}]$ the number of states is within $2^{O(r^2 \log n)}$ and therefore fits bound of \cref{thm:framework}.
More precisely, to ensure a clean interface with the other protocols, we assume the agents have field $\LeaderBit$ that indicates whether the agent is the leader and $\LeaderDone$ that indicates that the protocol has finished for this agent.
Altogether, we show the following:
\begin{lemma}[Fast Leader Election]
  There is a population protocol using $2^{O(\log n)}$ states electing a unique leader in time $O( \log n)$ w.h.p. starting from an awakening configuration.
  Further, after at most time $O(\log n)$, there is, w.h.p., an agent $\ell$ with $\LeaderDone = 1$ and $\LeaderBit = 1$, and at that time, w.h.p., all other agents $v \neq \ell$ have $\LeaderBit = 0$.
\end{lemma}
\begin{wrapfigure}{R}{0.55\textwidth}
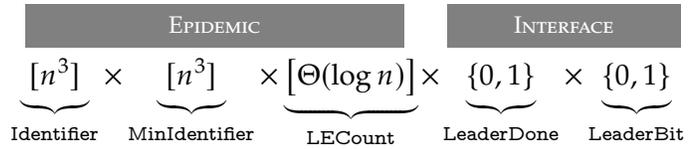

\centering
\[
\arraycolsep=0.5pt
\begin{array}{@{}ccccccccccccccccccc@{}}
{\graybox{12.5em}{\mathstrut\textsc{\footnotesize Epidemic}} }  & &{\graybox{7.5em}{\mathstrut\textsc{\footnotesize Interface}} }\\
  \underbrace{[n^3]}_{\Identifier}  \times  
\underbrace{[n^3]}_{\MinID}
\times
\underbrace{\left[\Theta(\log n)\right]}_{\LECount} &\times &
\underbrace{\{0,1\}}_{\LeaderDone} \times
\underbrace{\{0,1\}}_{\LeaderBit}
\end{array}
\]
    \caption{An overview of $\FastLeaderElect$'s state space. The overall state complexity is $2^{\Oh(\log{n})}$.}
    \label{fig:state-space-le}
\end{wrapfigure}
The state space of the protocol $\FastLeaderElect$ is presented in \cref{fig:state-space-le}.
In a nutshell, the protocol works as follows:
On its first activation (which either happens if the agent wakes up from $\PropagateReset$ or it interacts with another computing agent) it draws a value $x \in [n^3]$ (almost) u.a.r. and sets its field $u.\Identifier$ to $x$.
Then each agent determines the minimal identifier via a two-way epidemic. 
To this end, it uses a field $\MinID$ that stores the smallest value seen so far. In other words, upon an interaction between agents $u$ and $v$ they set
\begin{align}
    u.\MinID,v.\MinID \gets \min \{u.\MinID,v.\MinID \}
\end{align}
Each agent further maintains a counter $\LECount$ that is decreased by one on every interaction.
The counter is initialized to $c \log n$ where $c > 14$.
After $c\log n$ interactions, i.e. when the counter reaches $0$, the agent $u$ sets $u.\LeaderDone$ to $1$.
Then, it checks if it has the minium out of all identifiers it has seen thus far, i.e., it test $u.\Identifier =u.\MinID$.
 If so, it declares itself the leader and set $u.\LeaderBit$ to $1$.
Otherwise, $u.\Identifier =u.\MinID$, it set $u.\LeaderBit$ to $0$.

We now proves the protocols correctness.
First we note that the minimal identifier is unique w.h.p. as we draw the identifiers from such a large space. The probability that any two agent draw the same identifier is $\frac{2}{n^3}$. 
As there are $O(n^2)$ possible pairs of agents, the union bound implies that w.h.p. there are no two agents with the same identifier.
Recall that an agent decleares itself if its $\LECount$ reaches $0$ and its own $\Identifier$ matches $\MinID$.
Thus, only one agent will ever set its $\LeaderBit$ to be $1$ \textbf{if} by the time an agent's timer runs out, it knows the minimum.
Therefore, it remains to show that this event happens w.h.p.
We will first show the following auxilliary lemma:
\begin{lemma}
Let $t_0$ be the time step in which the first agent starts executing the protocol. Them, w.h.p., in time step $t_0 + 14 \cdot n \cdot \log n$ every agent $v \in V$ stores the minimum identifier in $v.\MinID$.
\end{lemma}
\begin{proof}
W.l.o.g. let $v_0 \in V$ be first agent that wakes up and starts executing the protocol and $v_{min} \in V$ be the agent that will draw the minimum identifer.
First, we determine the time $t'$ in which $v_{min}$ draws its identifier. 
Note that that any agent that interacts with an agent that executes the protocol will also start executing the protocol.
Thus, $t'$ must be smaller or equal time that an epidemic from $v_0$ takes to reach $v_{min}$.
Then, after another epidemic time, every agent must learned the minimum value. 
Therefore, after twice the epidemic time, all agents must have learned.
By \cref{lemma:epidemic} a single epidemic requires $7 \cdot n \cdot \log n$ interactions to reach all agents, w.h.p..
Using the union bound, two sequential epidemic therefore require $14 \cdot n \cdot \log n$ interactions, w.h.p.
Thus, the lemma follows.
\end{proof}
Recall that an agents counter $\LECount$ is decreased on every interaction.
By \cref{lem: full population technical} an agent's counter $\LECount$ will, w.h.p., not reach $0$ within $14 n \log n$ interaction for a large enough $c$. 
Thus, all agents will, w.h.p, receive the minimum value before they decide whether they are the leader or not. 
Together with the fact that the counter will hit $0$ in $O(\log n)$ times, this proves the lemma.

\section{Proofs and Additional Pseudocode for \texorpdfstring{$\DetectCollision$}{DetectCollision}}
\label{apx:error_detection}
In this section we present the omitted sub-protocols of $\DetectCollision$. Furthermore, we  formalise our requirements of the sub-protocol and provide a proof that it satisfies these conditions.
\begin{algorithm}[!ht]{$\CheckMessageConsistency(u,v)$\label{pr:consistency}}
for each $j \in [r_u^2]$ do
  if $v.\messages [(u.\rank , j)] \neq \bot$ and $v.\messages [(u.\rank , j)] \neq u.\observations [j]$ then
    set local state of $u,v$ to $\top$
\end{algorithm}

\begin{algorithm}[!ht]{$\UpdateMessages(u,v)$\label{pr:updatemsg}}
$u.\counter  \leftarrow u.\counter +1$ /*Update the counter*/
if $u.\counter =c\log{r_u}$ then
  $u.\tagVal \gets$  value chosen u.a.r. from $\{1,...,r_u^5\}$
  $u.\counter  \leftarrow 1$ /*Update the signature and reset the counter*/
  for each $j \in [r_u^2]$ do /*Update u's messages and observations to match the new signature*/
    if $u.\messages [(u.\rank ,j)]\neq \bot$ then
      $u.\messages [(u.\rank ,j)]\leftarrow u.\tagVal $
      $u.\observations [j]\leftarrow u.\tagVal $
for each $j \in [r_u^2]$ do /*Update v's messages that u governs to have u's signature*/
  if $v.\messages [(u.\rank ,j)]\neq \bot$ then
    $v.\messages [(u.\rank ,j)]\leftarrow u.\tagVal $
    $u.\observations [j]\leftarrow u.\tagVal $
\end{algorithm}

\newcommand{\TempArrayU}{\FieldName{uTemps}}
\newcommand{\TempArrayV}{\FieldName{vTemps}}
\newcommand{\IDs}{\FieldName{IDs}}
\newcommand{\FloorIDs}{\FieldName{floorIDs}}
\newcommand{\CeilIDs}{\FieldName{ceilIDs}}

\begin{algorithm}[!ht]{$\LoadBalance(u,v)$\label{pr:loadbalance}}
$\TempArrayU, \TempArrayV$ $\gets$ copies of $u.\messages, v.\messages$ with all entries set to $\bot$
for each $i \in \smap(u)$ do
  for each $k \in [1,r_u^5]$ do
    $\IDs \gets \{j \in [r_u^2]: u.\messages [(i,j)]=k \vee v.\messages [(i,j)]=k\}$
    $\FloorIDs, \CeilIDs \gets$ partition of $\IDs$ into two sets of size $\lfloor\frac{|\IDs|}{2}\rfloor$ and $\lceil\frac{|\IDs|}{2}\rceil$ respectively
    if $\TempArrayU$ has more non-$\bot$ cells than $\TempArrayV$ then
      for each $j \in \FloorIDs$ do $\TempArrayU[(i,j)]\leftarrow k$
      for each $j \in \CeilIDs$ do $\TempArrayV[(i,j)]\leftarrow k$
    else then
      for each $j \in \FloorIDs$ do $\TempArrayV[(i,j)]\leftarrow k$
      for each $j \in \CeilIDs$ do $\TempArrayU[(i,j)]\leftarrow k$
$u.\messages  \leftarrow \TempArrayU$
$v.\messages  \leftarrow \TempArrayV$
\end{algorithm}
We consider $\ErrorDetection$ as taking read-only inputs in $[n]$ (the $\rank$ field),
    with its additional state taken as a set $\QVerify$.
We will thus write agents' states as tuples in $[n] \times \QVerify$,
    and consider its transition function as mapping from $([n] \times \QVerify)^2$ to $\QVerify^2$.
The state $\top \in \QVerify$ indicates that an error was found.
Further, we say that the message indexed by $(i,j)$ is governed by agents of rank $i$ and has $j$ as its ID. Additionally, we say that agent $u$ has message $(i,j)$ with content $k$ if the cell of $u$'s $\messages $ indexed by $(i,j)$ has the value $k$.

We show that $\ErrorDetection$ will never generate such a $\top$ when correctly initialized on a correct ranking,
    but quickly generate such a state under any initialization on an incorrect ranking (w.h.p.).

\begin{lemma}
    \label{prop: error detection}
    The protocol $\ErrorDetection$ (with parameter $1 \leq r < n/2$) uses $2^{\Oh(r^2 \log n)}$ states and has the following properties:\begin{enumerate}[label=(\alph*),noitemsep]
        \item \emph{Soundness: No false positive upon correct initialization on correct ranking.}
            Let $(x_i)_{i \in [n]} \in [n]^n$ be a correct ranking (i.e., a permutation of $[n]$).
            Then from the (correct, initial) configuration $\vec{c}' = ((x_i, \qNaughtDetectCollision))_{i \in [n]}$,
                no configuration with an agent having a state $(\cdot, \top)$ is reachable by $\ErrorDetection$.
        \item \emph{Robust completeness: Fast detection of erroneous rankings, no matter $\ErrorDetection$'s state.}
            If at time $t$, the configuration $((x_i, q'_i))_{i \in [n]} \in ([n] \times \QVerify)^n$
                is such that $(x_i)_{i \in [n]}$ is not a correct ranking,
                then within $\Oh(\frac{n^2 \log n}{r})$ interactions,
                    there will be an agent in a state $(\cdot, \top)$.
    \end{enumerate}
\end{lemma}
\subsection{Proof of \texorpdfstring{\cref{prop: error detection}}{Main Lemma for DetectCollision}}

We will show that the error-detection module correctly generates a $\top$ within $O(\nicefrac{n^2}{r}\log{n})$ interactions w.h.p.\ if there exists a duplicate rank and will never do so if the ranking is correct and all agents are initialised from $\qNaughtDetectCollision$.
In fact we will prove a slightly stronger version necessary for the lower memory regime as we require that the module is able to function even when there are more agents than expected in a given group.
To make this clear we will adjust notation, we have that an instance of the module is designed assuming there are $m$ agents in the group, however we assume there are in fact $\eta$ for $m<\eta$.
Evidently, we cannot have a correct ranking if $m<\eta$ and so in these cases we are only interested in generating a $\top$ character.

We begin by making the following observations:
\begin{enumerate}[noitemsep]
    \item The rank of an agent cannot change
    \item Messages are never created or destroyed
    \item If the protocol is begun from a configuration with a correct ranking and all collision detection states $\qNaughtDetectCollision$ there is only ever one message with index $(i,j)$.
    \item Unless it is updated by an agent of rank $i$ the content of a message with index $(i,j)$ does not change
    \item If the content of message $(i,j)$ is updated the contents of the corresponding observation in its governor's memory is updated to match and the inverse is also true.
\end{enumerate}
\begin{lemma}
    \label{lemma: error-detection no fps}
    If the error-detection module generates a $\top$ it did not begin from a correct ranking with all collision detection states set to $\qNaughtDetectCollision$.
\end{lemma}
\begin{proof}
    Assume we began from a correct execution, then there are three ways that the error-detection module can generate a $\top$. The first is that there are two agents with the same rank, which immediately violates either the assumption or the first observation. Second, there could be two messages with the same ID. However, this violates the second and third observations. Finally, the contents of a message doesn't match a corresponding observation. By the fourth and fifth observations, this requires that there are two agents with the same rank, which violates our assumption. Thus, we obtain a contradiction.
\end{proof}
It now remains to show that from any configuration with duplicate ranks a $\top$ is generated sufficiently quickly. 

We will first restrict ourselves to considering the case where there is only one rank with multiple agents. In this case, we can immediately obtain our desired $\Oh(n\log{n})$ interactions in the case where there are sufficiently many duplicates.
\begin{lemma}
    If there are at least $m$ agents with non-unique ranks, the protocol generates a $\top$ within $\Oh(m\log{m})$ interactions w.h.p.
    \label{lemma: Reset quickly for at least root n duplicates.}
\end{lemma}

\begin{proof}
     Within the setting of this lemma, the total number of agents with non-unique ranks is $\max(m,\eta-m)$. A $\top$ will be generated the first time two agents of the same rank meet. Considering the worst case where we must wait for this to occur, this event occurs with probability at least $\nicefrac{1}{2m}$ whenever the scheduler generates an interaction. Let $\NoResetBy{t}$ be the event that a $\top$ has not been generated by interaction $t$. With at least $m$ non-unique agents,
    \[\Prob{\NoResetBy{6m\log{m}}}\leq \left(1-\frac{1}{2m}\right)^{6m\log{m}}\leq m^{-3}=o(m^{-2}).\qedhere\]
\end{proof}

Thus for the case where there are more than $m$ agents with non-unique ranks, or $\eta>2m$ we have our claim. However, for the case where there are fewer than $m$ agents with non-unique ranks we must rely on the other methods to generate a $\top$. Primarily this will mean a message being modified by multiple agents and causing a conflict between the contents of the message and an agent's stored observation. We will begin with a technical lemma where we define $\Interactions{u,t_0,x}$ to be a random variable recording the number of interactions agent $u$ is involved of the first $x$ interactions following time $t_0$, and assume $2m<\eta$.
\begin{lemma}
    For constant $a>16$ and all $t\geq 0$, $\Prob{\frac{a}{2}\log{m}\leq \Interactions{u,t,am\log{m}}\leq 4a\log{m}}\geq 1-o(m^{-2})$.
    \label{lemma: technical interaction bound}
\end{lemma}
\begin{proof}
    Since, $\Interactions{u,t,t+am\log{m}}\sim \Bin(am\log{m},\frac{2}{\eta})$, we have $\Exp{\Interactions{u,t,am\log{m}}}=\frac{2am\log{m}}{\eta}$. For the lower bound we apply the weak multiplicative form of the Chernoff bound for sum of i.i.d.\ Bernoulli variables and take $\eta=2m$.
    \[\Prob{\Interactions{u,t,am\log{m}}\leq \frac{1}{2} \Exp{\Interactions{u,t,am\log{m}}}}
        \leq \exp\bc{-\bc{\frac{1}{2}}^3\Exp{\Interactions{u,t,am\log{m}}}}
        = \exp\bc{-\frac{a\log{m}}{8}}
        =m^{-\frac{a}{8}}.\]
    For the upper bound, we apply the corresponding bound, but instead take $\eta=m$
    \[\Prob{\Interactions{u,t,am\log{m}}\geq 2\Exp{\Interactions{u,t,am\log{m}}}} 
        \leq \exp\bc{-\frac{\Exp{\Interactions{u,t,am\log{m}}}}{3}}
        =m^{\frac{-2a}{3}}.\]
    Thus for $a>16$ the claim holds.
\end{proof}
\begin{lemma}
    From any configuration of agents in the collision detection module such that agent $u$ has rank $i$ and there are at most $m-1$ other agents with rank $i$. Within $cm\log{m}$ interactions, $u$ will adopt a unique $\tagVal $ among all agents of rank $i$, messages governed by rank $i$ and observations of those messages, retain this signature for $\frac{c}{8}m\log{m}$ interactions and no other agent of rank $i$ will adopt this signature within $\frac{c}{8}m\log{m}$ interactions w.h.p.\ for $c>32$.
    \label{lemma:Tags are unique}
\end{lemma}
\begin{proof}
    We have three failure cases: the agent does not adopt a unique signature within $cm\log{m}$ interactions, the agent loses this signature before $\frac{c}{4}m\log{m}$ interactions or another agent adopts the signature within that time. 
    The first means either the agent does not update its signature, or it does not pick a unique one. 
    The former implies it has activated less than $c\log{m}$ times which has probability $o(m^{-2})$ by \cref{lemma: technical interaction bound} and the latter occurs with probability $o(n^{-1})$ as a signature is chosen uniformly at random from a space of size $m^{5}$ and there are $4m^3+m=\Oh(m^{3})$ possible collisions. 
    For the agent to lose its signature before $\frac{c}{8}m\log{m}$ interactions implies that the agent interacted at least $c\log{m}$ times during that window, however by \cref{lemma: technical interaction bound} we have that this occurs with probability $o(m^{-2})$. 
    Finally, for another agent of rank $i$ to adopt the same label requires that agent to refresh its signature and pick the right label. 
    The probability that any given agent $v$ does this is upper bounded by the probability that it picks correctly the first time it updates or it gets to update twice within the period. 
    Picking the correct signature occurs with probability $m^{-5}$ and updating twice within that period requires $v$ to interact at least $c\log{m}$ times which by the previous lemma occurs with probability at most $o(m^{-2})$. 
\newcommand{\Behaves}[1]{\ensuremath{\mathsf{Behaves}_{#1}}}
    Let $\Behaves{u}$ be the event that all of the events in the lemma statement hold, then by a union bound
    \[\Prob{\neg \Behaves{u}}
       \leq  o(m^{-2})+o(m^{-1})+o(m^{-2})+m(o(m^{-2})+m^{-5})
       =  o(m^{-1}).\qedhere\]
\end{proof}

The above result implies that within $\Oh(m\log{m})$ interactions we modify a large number of messages such that they could trigger the generation of a $\top$. 
We now must show that they will actually be seen by other agents of the rank that governs them and so cause a $\top$ to be generated. 
To do this we use load balancing to disperse the modified messages and inherit a nice result from the load balancing literature. We begin by considering an idealized version of the process.

Take an idealized version of the module with the following additional properties:
\begin{itemize}
    \item the contents of messages and observations do not change, and 
    \item $\top$ characters are never generated.
\end{itemize}
Let $X=(X_t)_{t\geq0}$ be a Markov chain over $\mathbb{N}^\eta$ such that the $i$th entry of $X_t$ is equal to the number of messages governed by $i$ with contents $k$ at agent $l$ after the $t$th interaction. 
Furthermore, let $Y=(Y_t)_{t\geq0}$ be a Markov chain over $\mathbb{N}^\eta$ be the Markov chain such that the $l$th entry of $Y_t$ is equal to the number of tokens at agent $l$ in the load balancing algorithm used in Tight and Simple Load Balancing \cite{berenbrink2019tight} at time $t$.

\begin{lemma}
    if $\eta\leq 2m$, there exists $t=\Oh(m\log{m})$ such that $X_t$ contains no zeros w.h.p.\ if we begin with exactly $4m$ messages governed by $i$ and with content $k$.
    \label{lemma: tight and simple}
\end{lemma}
\begin{proof}
    We begin by showing a coupling between $X$ and $Y$ with initial configurations such that $X_0=Y_0$ (i.e. the number of tokens at each agent in $Y$ match the number of messages governed by $i$ with content $k$ at each agent in $X$). 
    The two processes are then determined entirely by the pairs the scheduler chooses to interact. 
    We construct the coupling as follows:
    If the scheduler for $X$ picks the pair $(u,v)$ and our idealized protocol assigns more messages governed by $i$ and with content $k$ to $u$, then the scheduler for $Y$ picks $(u,v)$. Otherwise, if our idealized protocol assigns more messages governed by $i$ and with content $k$ to $v$, the scheduler for $Y$ picks $(v,u)$. 
    The former case leads to $u$ having $\lceil\frac{(\#u+\#v)}{2}\rceil$ of the messages and $v$ having $\lfloor\frac{(\#u+\#v)}{2}\rfloor$ in both $X$ and $Y$, while the latter gives the inverse (where $\#w$ is the number of messages at $w$ governed by $i$ and with content $k$.) 
    Furthermore, as for all ordered pairs of agents the probability of being chosen by the scheduler for $X$ is $\frac{1}{\eta(\eta-1)}$ and at each point there is a bijection between the ordered pairs in the scheduler of $X$ and the scheduler for $Y$, we must have that the probability of any given ordered pair being chosen by the scheduler of $Y$ is also $\frac{1}{\eta(\eta-1)}$. 
    Thus, the marginal distributions are correct and at all times $X_t=Y_t$.

    We then apply theorem 1 of Tight and Simple Load Balancing \cite{berenbrink2019tight}, to obtain the fact that from an initial discrepancy of at most $2\eta$ all agents must have either $1$, $2$ or $3$ tokens in $Y$ after $\Oh(m\log{m})$ interactions with probability $1-m^{-\Omega(1)}$. 
    By repeated application of this bound over a constant number of $\Oh(m\log{m})$ intervals, we can then amplify the probability of this occurring to be at least $1-m^{-2}$. This gives our result.
\end{proof}
It is clear to see that adding additional messages governed by $i$ and with content $k$ does not cause the previous lemma to fail. However, in order for this result to apply to our non-idealized process, we need to deal with $\top$ being generated and messages governed by $i$ and with content $k$ being removed. In the latter case this requires an agent with rank $i$ to modify it and this implies either the original agent with signature $k$ has changed its signature or the message arrived at a duplicate which would have generated a $\top$.
\begin{lemma}
    There exists a $c>0$, such that from any initial configuration where all agents are in the error correction module, if there are multiple agents with the same rank then a $\top$ is generated within $\Oh(m\log{m})$ interactions with probability $1-o(m^{-1})$.
    \label{lemma:error-correction theorem}
\end{lemma}
\begin{proof}
    While there may be multiple ranks with multiple agents, we will instead pick one such rank $i$ arbitrarily. Assume there are  duplicate agents with rank $i$. If there are at least $m$ agents with non-unique ranks we obtain the result immediately from  \autoref{lemma: Reset quickly for at least root n duplicates.}. Otherwise, we present the following sequence of events and show that it occurs or a $\top$ is generated  w.h.p.\ within $\Oh(m\log{m})$ interactions: an agent with rank $i$ picks a unique signature $k$ and modifies at least $4m$ messages to have contents $k$, then one of the duplicates receives a message governed by $i$ with contents $k$ and generates a $\top$.
    For this not to occur one of the following must hold, the event described in  \autoref{lemma:Tags are unique} does not occur for any of the agents of rank $i$ or the load balancing fails. These events both occur with probability at most $o(m^{-2})$ and so we have that the above sequence of events occurs with probability at least $1-o(m^{-2})$ by a union bound. Thus, w.h.p., a $\top$ will be generated and the claim holds.
\end{proof}

If $r=\Theta(n)$, we have obtained correctness and the desired run time. However, we still need to show correctness in the lower memory regime.
\begin{proof}[Proof of \cref{prop: error detection}]
    We will first show that $\ErrorDetection$ is sound.
    Clearly, here exists a subset in some partition of $[n]$ containing a duplicate rank if and only if there is a duplicate for some rank in $[n]$.
    Thus, given that $\ErrorDetection$ is sound when applied to the whole population (\cref{lemma: error-detection no fps}), it must be sound when applied to one of the subsets.

    For robust completeness, let $F=\{S_1,...,S_{\nicefrac{n}{r}}\}$ be a partition of $[n]$ such that each $|S_i|=r$ and denote by $A(S_i)$ the set of agents with a rank from $S_i$. 
    If there is a duplicate rank, there must exist $i \in [\nicefrac{n}{r}]$ such that $|A(S_i)|\geq r$ and $A(S_i)$ contains a duplicate rank, from here on we will denote one such $S_i$ by just $S$.
    Applying \cref{lemma:error-correction theorem}, we must have that the error-detection module associated with $S$ generates a $\top$ within $\Oh(r\log{r})$ interactions between agents of $A(S)$ regardless of initialization w.h.p. in $r$.
    By considering $\nicefrac{\log{n}}{\log{r}}$ such consecutive intervals, we can amplify this to a $1-o(n^{-1})$ probability to generate a $\top$ within $r\log{n}$ interactions between agents of $A(S)$.

    Now consider that, the probability that any interaction of the whole population occurs between two agents of $S$ has probability $\frac{|A(S)|^2-|A(S)|}{n(n-1)}<\nicefrac{4r^2}{n^2}$. 
    Since, each pair is selected independently, the number of interactions between pairs in $A(S)$ within $t$ interactions of the overall population is distributed binomially. 
    Thus, by applying a Chernoff bound, we have that there exists constant $c_{slowdown}>0$ such that for $t$ sufficiently large there have been at least $\frac{4r^2}{n^2}t$ interactions between agents in $A(S)$ within $c_{slowdown}t$ interactions of the overall population w.h.p. in $n$.
    Combining results, we have that if there exists a duplicated rank then a $\top$ is generated within $\Oh(\nicefrac{n^2}{r}\log{n})$ interactions w.h.p. 
    This gives the claim.
\end{proof}

\section{Deferred proofs from \texorpdfstring{\cref{sec:analysis}}{our Analysis}}
In this section, we provide the previously deferred proofs for the results contained in \autoref{sec:analysis}.
As in that section, we split the results into stability/safety, correctness after a reset and recovery.

\subsection{Stability: Proof of \texorpdfstring{\cref{lem:safe_set_is_safe}}{Lemma 6.1}}\label{apx:lem:safe_set_is_safe}
    Let $C \in \safe$.
    We consider an interaction between two agents in $C$, and let $C'$ be the resulting configuration.
    We show that for each possible combinations of agents' $\generation$s,
        $C' \in \safe$.

    Since all agents are verifiers, the only way for a non-verifier to be created or for the ranking to change is for $\top$ to be generated at an agent with a non-zero probation timer.
    If the two interacting agents are both in generation $i$, they must both have a $\probationTimer$ of $0$ and so they cannot violate this.
    If one agent is in generation $i$ and one agent is in generation $i+1$, the agent in generation $i$ advances and adopts $\qNaughtDetectCollision$ as its $\ErrorDetection$ state.
    This corresponds to an interaction of $\tilde{C}_{i+}$ and so by assumption it is reachable from $\tilde{C}_0$ by $\ErrorDetection$.
    Thus, it cannot generate a $\top$ by the \soundness{} of collision detection \autoref{prop: error detection}.
    Similarly, if both agents are in generation $i+1$, this is also corresponds to an interaction of $\tilde{C}_{i+}$ and so a $\top$ cannot be generated.
    Therefore, $C'$ also has all agents being verifiers and the same correct ranking.

    As we have established that only agents in generation $i$ can generate a $\top$ and the only generations present are $i$ and $i+1$, it follows that only agents in generation $i$ can change generation.
    Furthermore, they can only advance to generation $i+1$.
    Thus, all agents remain in generations $i$ and $i+1$.
    Additionally, as the only way for an agent's probation timer to increase is for it to change generation, no agent in generation $i$ can obtain a non-zero probation timer.
    The act of doing so would necessitate them leaving generation $i$.
    Therefore, $C'$ also has all agents in either generation $i$ or $i+1$, and all agents in generation $i$ must have probation counter of $0$.

    Finally, we consider the reachability of $\tilde{C}'_{i+}$ from $\tilde{C}_0$ by \ErrorDetection.
    If one or more of the agents are in generation $i+1$, as previously established, the resulting interaction is exactly one that could occur from $\tilde{C}_{i+}$ and so by the transitivity of reachability $\tilde{C'}_{i+}$ must also be reachable from $\tilde{C}_{0}$.
    Otherwise, both agents must be from generation $i$.
    Either they both remain in generation $i$ or both advance to generation $i+1$, resetting their collision detection states to $\qNaughtDetectCollision$ and so $\tilde{C}'_{i+}=\tilde{C}_{i+}$.
    Thus, $\tilde{C}'_{i+}$ is reachable from $\tilde{C}_{0}$ by \ErrorDetection.

    Therefore, if $C$ is safe $C'$ is also safe.

\subsection{Correctness after a full reset: Proof of \texorpdfstring{\cref{lem:correctness_from_reset}}{Lemma 6.2}}
\label{apx:correctness after reset}
In this subsection, we prove \cref{lem:correctness_from_reset}, showing that the protocol stabilizes to a configuration from $\safe$ quickly following a reset w.h.p.
To achieve this,
    in addition to some lemmas proven in preceding appendices,
    we require the following lemma, which handles the correct transition from a correctly computed ranking into $\safe$.

\begin{lemma}
\label{thm:stable_verify_correctness}
    Let $t$ be a time such that all agents in the population have $\role=\producing$, have $\countdown > 0$, where $(u.\qRanking.\rank)_{u \in [n]}$ is a correct ranking, and where the configuration $(u.\qRanking)_{u \in [n]}$ is a silent configuration for $\Ranking$.
    Then $\ElectLeader$ reaches a configuration in $\safe$ within $\Oh(\nicefrac{n^2}{r} \log n)$ interactions.
\end{lemma}

\begin{proof}
    From time $t$, as long as all $\countdown$ fields are larger than $0$,
        the only thing that happens is that agents decrement the $\countdown$ on every interaction.
    As the maximum value of $\countdown$ is $C_{max} = \Oh(\nicefrac{n}{r} \log n)$, by \cref{lem: full population technical},
        within $\Oh(\nicefrac{n^2}{r} \log n)$ interactions, some agent's $\countdown$ value reaches $0$.
    At that time, that agent becomes a verifier.
    This spreads through the population as an epidemic (possibly with spontaneous infections as other agent's independently have their $\countdown$ hit $0$).
    Thus, after at most $\cepi n\log{n}$ further interactions of the overall population, every agent either is a verifier with the same rank as at the beginning of the epidemic, or a reset was triggered, with probability $1-o(n^{-2})$.
    However, as we will demonstrate, the reset is not triggered w.h.p., such that protocol enters a safe configuration w.h.p.

    Upon first becoming a verifier, all agents have $\generation=0$ and $\qVerifying=\qNaughtDetectCollision$.
    For a reset to occur, an agent must attempt to change its $\generation$.
    As all agents begin in generation $0$, this means that a $\top$ must be generated by $\ErrorDetection$ when two verifiers of the same generation interact.
    When a ranker meets a verifier, the ranker becomes a verifier with $\qVerifying$ initialized to $\qNaughtDetectCollision$.
    As no agent changes its $\qRanking$ field before becoming a verifier, the sequence of interactions between exclusively verifiers corresponds exactly to a sequence that could occur from the configuration of $\ErrorDetection$ obtained by initializing it's state to $\qNaughtDetectCollision$ on the ranking at the time of the first agent becoming a verifier.
    By \cref{prop: error detection}, no such sequence of interactions can generate the state $\top$, and hence no reset can be triggered.

    Once all agents enter the verifier state, the configuration is in $\safe$ and thus safe.
    Therefore, with probability $1-o(n^{-2})$ within $\Oh(\effn)+\cepi n\log{n}=\Oh(\effn)$ interactions of an awakening configuration the population reaches a safe configuration.
\end{proof}
With that established we present the following proof.
\begin{proof}[Proof of \cref{lem:correctness_from_reset}]
From a partially triggered configuration (i.e., one where a reset has just been called), by \autoref{lemma:propagatereset} we obtain a fully-dormant configuration within $\Oh(n\log{n})$ interactions w.h.p.
From a fully dormant configuration, either some agent has left $\ranking$ or by \autoref{lem:fast-and-silent-ranking}, within $\crank\nicefrac{n^2}{r}\log{n}$ interactions there exists a configuration where all agents have a unique rank and our ranking sub-protocol has become silent w.h.p.
From a configuration where all agents are dormant and $\resetting$, or $\ranking$ the only way for the first agent to leave $\ranking$ is for its $\countdown$ field to hit $0$.
However, for it to have run out before $\crank\nicefrac{n^2}{r}\log{n}$ interactions have taken place, would require some agent to interact $C_{\text{max}}\nicefrac{n}{r}\log{n}$ times within that same period.
If we take $C_{\text{max}}>>\crank$, by application of \autoref{lem: full population technical}, this does not occur w.h.p.
From this configuration where all agents are in $\ranking$, have unique ranks and the ranking sub-protocol has become silent, all agents reach a configuration in $\safe$ within $\Oh(\nicefrac{n^2}{r}\log{n})$ interactions by \autoref{thm:stable_verify_correctness}.
\end{proof}

\subsection{Recovery: Deferred lemmas for \texorpdfstring{\cref{lem:recovery}}{Lemma 6.3}}\label{apx:lem:recovery}
In this subsection we present the proof of \cref{lem:recovery}, showing that from an arbitrary configuration the protocol reaches the safe set quickly or a reset occurs w.h.p.
Recall the hierarchy of configuration sets we defined in the text of the proof of \cref{lem:recovery}:
$\cC \eqqcolon \cC_0 \supset \cC_1 \supset \cC_2 \supset \cC_3 \supset \cC_4 \supset \cC_5$ is a hierarchy of configuration sets containing exactly those configurations where the following properties are true:
\begin{itemize}[noitemsep]
    \item $\cC_1$: All agents are rankers or verifiers (i.e., no agent is a resetter).
    \item $\cC_2$: All agents are verifiers.
    \item $\cC_3$: Additionally, all $\generation$ values are equal.
    \item $\cC_4$: Additionally, all agents have $\probationTimer = 0$.
    \item $\cC_5$: Additionally, all $\rank$ fields are all distinct, i.e., the ranking is correct.
\end{itemize}

\begin{lemma}\label{lem:recovery:resetting}
    Starting from a configuration in $\cC \setminus \cC_1$, $\ElectLeader$ reaches a configuration with no agents with role $\resetting$ within $\Oh(n\log{n})$ interactions w.h.p.
\end{lemma}

\begin{proof}
    It follows from \cref{lemma:propagatereset} that unless a reset is triggered, no agent will be $\resetting$ within $\Oh(n\log{n})$ interactions w.h.p.\
    If a reset is triggered within this time, w.h.p. within $\Oh(n\log{n})$ further interactions there will be a fully dormant configuration.
    From a fully dormant configuration, w.h.p. all agents wake up within a further $\Oh(n\log{n})$ interactions or a reset occurs.
    However, for a reset to occur before all agents awaken requires an agent's $\countdown$ field to hit $0$.
    This in turn requires some agent to have taken part in at least $C_{\text{max}} \cdot \nicefrac{n}{r}\log{n}$ interactions within this period, but by taking $C_{\text{max}}$ sufficiently large, we find this does not occur w.h.p. by application of \autoref{lem: full population technical}.
    Thus, w.h.p. within $\Oh(n\log{n})$ interactions of a reset there exists a configuration containing only $\ranking$ agents.
    Therefore, the statement holds.
\end{proof}

\begin{lemma}\label{lem:recovery:producing}
    Starting from a configuration in $\cC_2 \setminus \cC_1$,
        $\ElectLeader$ triggers a reset or reaches a configuration in $\cC_2$ within $\Oh(\nicefrac{n^2}{r}\log{n})$ interactions w.h.p.
\end{lemma}

\begin{proof}
    By assumption, the starting configuration contains no resetters, so only rankers and verifiers.
    Thus, for any agent to become a resetter, a reset needs to be triggered.
    As a ranker may become a verifier, but not vice versa (barring a reset),
        the number of verifiers in monotonically increasing until a reset is triggered.
    After at most $\Oh(\effnovern)$ interactions as a ranker, an agent becomes a verifier (because then its $\countdown$ field has reached $0$, if it is not turned into a verifier before then).
    By \cref{lem: full population technical},
        w.h.p., all agents interact $\Omega(\effnovern)$ times within $\Oh(\effnovern)$ interactions.
    Thus, within $\Oh(\effn)$ interactions, w.h.p., either a reset has occurred or all agents are verifiers.
\end{proof}

\renewcommand{\echo}{\textsf{EchoChamber}}

\begin{lemma}\label{lem:recovery:generations}
    Starting from a configuration in $\mathcal{C}_2\setminus \mathcal{C}_3$,
        the protocol either triggers a reset or reaches a configuration in $\cC_3$ within $\Oh(n \log n)$ interactions w.h.p.
\end{lemma}

\begin{proof}
    Recall that in a configuration in $\cC_2 \setminus \cC_3$ all agents are verifiers, but there are agents with different $\generation$ fields;
        we need to show that the $\generation$s equalize or that a reset is triggered quickly w.h.p.
    To that end, we show the following claims in turn, which taken together imply the lemma.
    \begin{enumerate}[noitemsep]
        \item From a configuration in $\mathcal{C}_2\setminus \mathcal{C}_3$, within $\Oh(n \log n)$ interactions w.h.p., either the population enters a configuration from $C_2$ where all agents belong to at most three generations, all of which are consecutive, or a reset is triggered. \label{claim: 3 generations}
        \item From a configuration in $\mathcal{C}_2$ where all agents are in at most three generations $\{i-2, i-1, i\}$, within $\Oh(n \log n)$ interactions w.h.p., either the population enters a configuration where all agents are in generation $i$ or $i+1$ and all agents in $i+1$ have a probation timer of at least $8\cepi \log{n}$, or a reset is triggered. \label{claim: Converge to two generations}
        \item From a configuration in $C_2$ where all agents are in at most two consecutive generations $i$ and $i+1$, and all agents in generation $i+1$ have a probation timer of at least $8\cepi\log{n}$, the population enters a configuration where all agents are in the same generation or a reset is triggered within $\Oh(n\log{n})$ interactions w.h.p. \label{claim: Converge to one generation}
    \end{enumerate}    
    
    For claim $\ref{claim: 3 generations}$, first notice that all $n$ agents belong to one of the $6$ generations.
    Hence, there is a generation $i$ to which $\geq n/6$ agents belong.
    If there is an agent belonging to a generation differing from $i$ by at least $2$,
        there are at least $n/6$ pairs of agents such that an interaction between them will trigger a reset.
    For the claim not to hold, the protocol must remain in such configurations for at least $cn\log{n}$ interactions for $c>6$. Since each interaction pair is sampled uniformly and independently at random, 
   the event that no pair of agents from non-adjacent generations interact for $cn\log{n}$ interactions can be upper bounded by
    \begin{equation}
     \bc{1-\frac{1}{3(n-1)}}^{cn\log{n}}
        \leq \bc{1-\frac{1}{3n}}^{cn\log{n}}
        \leq n^{-\nicefrac{c}{3}}
        =o(n^{-2})
    \end{equation}

    \noindent
    For claim \ref{claim: Converge to two generations}, we require the following two observations.
    \begin{observation}
    \label{apx: lemma observations}
        W.h.p.\ either each agent changes their generation only once within $\nicefrac{n}{8}\cdot P_{max}$ interactions or a reset is triggered.
    That is because an agent can change its generation only via a soft reset,
        which also resets the agent's probation counter to $P_{max}$,
        so that it will not reach $0$ in that time from w.h.p. (by \cref{lem: full population technical}).
    And while an agent is on probation, it will not perform soft, but full resets.
    \end{observation}
    \begin{observation}
        An agent in generation $m$ can only move to generation $m+1$.
    \end{observation}
    Thus from the configuration in $\mathcal{C}_2$ with all agents in a generation from $\{i-2,i-1,i\}$ within $\nicefrac{n}{24}\cdot P_{max}$ interactions either a reset occurs or all agents belong to a generation from $\{i-2, i-1, i, i+1\}$ w.h.p.
    When two agents interact, either the agents both join the maximum of their two generations, the agents spontaneously advance a generation, or a reset occurs.
    We observe that this is similar to the behaviour of a max epidemic and we shall construct a coupling to make this more explicit.
    
    Let $t$ be the first time the agents enter a configuration from the statement of Claim \ref{claim: Converge to two generations} and let $G_t=(g_{j})_{j \in [n]}$ be the generation numbers of the agents in the configuration at that time. We now define two processes both starting in  configuration $G_0$.
    \begin{itemize}[noitemsep]
        \item For $t' > t$, $G_{t'}=(g_{j,t'})_{j\in[n]}$ are the generation numbers of all agents in the configuration of the protocol at time $t'$.
        \item Further, take $\hat{G}_{t'}=(\hat{g}_{j,t'})_{j \in [n]}$ to be the generation numbers of all agents in the configuration at time $t' > t$ calculated as follows:
            after time $t$, when two agents interact, they just adopt the maximum generation\footnote{Because at this point we are in configurations where all occupied generations lie within a set of three consecutive generations, this is well-defined.} (i.e., they perform a max epidemic).
    \end{itemize}
    The processes are coupled via identity coupling (the same pair interacts in both processes).
    After each interaction, for each of the two processes the following holds:
    \begin{itemize}[noitemsep]
        \item either a reset has been triggered, 
        \item an agent has advanced to generation $i+2$, or 
   \item  all agents in the original process have a generation at least as large $(\bmod 6)$ as they do in the max-epidemic version.
   \end{itemize}
   We wish to show that after $\nicefrac{n}{24}\cdot P_{max}$ interactions, either a reset has occurred for $G_{t'}$, a reset has occured in both processes,  or all agents are in generation $i$ or $i+1$.

   By our observations (see \ref{apx: lemma observations}) w.h.p. the following holds: if an agent has advanced to generation $i+2$ in $G_{t'}$, the process will have performed a reset  within that time.
   Furthermore, since $\nicefrac{n}{24}\cdot P_{max} \gg \cepi n\log{n}$ the max-epidemic will have succeeded, meaning w.h.p. all agents of $\hat{G}$ will be in generation $i$. To summarize, 
   w.h.p. either a reset has occurred or all agents of $G_{t'}$ are in generation $i$ or $i+1$.
    Furthermore, any agent in generation $i+1$ must have advanced within the last $\nicefrac{n}{24}P_{max}$ interactions. 
    Therefore, any agent in generation $i+1$ has either taken part in $P_{max}-8\cepi\log{n}$ interactions within the $\nicefrac{n}{24}P_{max}$ interactions or has a $\probationTimer$ of at least $8\cepi \log{n}$.
    This gives the claim.
    By application of \cref{lem: full population technical} we find that this has occurred for any agent has probability $o(n^{-2})$.
    Thus,  Claim \ref{claim: Converge to two generations} holds.

    It remains to show Claim 3. If all agents are in generation $i$ or all agents are in generation $i+1$ claim \ref{claim: Converge to one generation} is immediate.
    If not, then a near identical argument as the one for Claim \ref{claim: Converge to two generations} holds.
    The highest generation number will w.h.p. spread by an epidemic which will conclude before the probation counter of any agent in generation $i+1$ runs out, permitting it to advance to generation $i+2$ without triggering a reset.
    Thus, w.h.p.\ either $i+1$ remains the highest generation and all agents belong to it or a reset is called within $\cepi n\log{n}$ interactions.
    Thus, Claim \ref{claim: Converge to one generation} holds, which finishes the proof of the lemma.
\end{proof}

\begin{lemma}\label{lem:recovery:probation}
    Starting from a configuration in $\mathcal{C}_3\setminus \mathcal{C}_4$,
        the protocol either triggers a reset or reaches a configuration in $\cC_4$ within $\Oh(\nicefrac{n^2}{r} \log n)$ interactions w.h.p.
\end{lemma}
\begin{proof}Recall that in configurations in $\cC_3 \setminus \cC_4$, all agents are verifiers and belong to the same generation.
    Hence, until a reset occurs, each agents decrements its $\probationTimer$ in every interaction, if possible.
    So by \cref{lem: full population technical}, all probation timers will reach $0$ or a reset is triggered within $\Oh(\effnovern)$ interactions.
    If there was no reset and the output is correct, the configuration is now safe.
    However, if an agent generates $\top$ before this occurs,
        it must either trigger a reset immediately or advance a generation. 
    News of this new generation spreads as a broadcast and so within $\Oh(n\log{n})$ interactions either all agents are in the new generation or a reset is triggered with probability $1-o(n^{-2})$.
    Upon entering the new generation each agent resets their probation timer to its maximum value.
    If the output is correct, the configuration is now safe.
    If the output is not correct, w.h.p.\ the verifier must produce $\top$ within $\Oh(\effn)$ interactions of the overall population by \cref{prop: error detection}.
    However, by \cref{lem: full population technical} and choice of $\cprobation$, with probability $1-o(n^{-2})$ no agents' $\probationTimer$ will have reached $0$ at that time, so a reset would then be triggered as all agents have probation timers greater than $0$.
    Thus, the claim holds.
\end{proof}

\begin{lemma}\label{lem:recovery:incorrect_output}
    Starting from a configuration in $\mathcal{C}_4\setminus \mathcal{C}_5$,
        the protocol triggers a reset within $\Oh(\nicefrac{n^2}{r} \log n)$ interactions w.h.p.
\end{lemma}
\begin{proof}
    By \cref{prop: error detection} a $\top$ must be generated within $\Oh(\effn)$ interactions w.h.p.
    This will cause the generating agent to advance to the next generation, which spreads via broadcast.
    Thus, after $O(n\log{n})$ interactions, w.h.p., all agents will be in the new generation or a reset has been triggered.
    If all agents enter the new generation, then by \cref{prop: error detection}, a $\top$ is again generated within $\Oh(\effn)$ interactions.
    However, by application of \Cref{lem: full population technical} and choice of $\cprobation$, w.h.p., all agents will have a non-zero $\probationTimer$ until this happens, so that a reset is triggered.
\end{proof}

\newpage
\printbibliography

\end{document}